\documentclass[11pt,letterpaper]{article}
\usepackage{amsmath,amsfonts,amsthm,amssymb}
\usepackage{setspace}
\usepackage{fancyhdr}
\usepackage{lastpage}
\usepackage{extramarks}
\usepackage{xspace}
\usepackage{chngpage}
\usepackage{algorithm,algorithmicx}
\usepackage[noend]{algpseudocode}
\usepackage{soul,color}
\usepackage{graphicx,float,wrapfig}
\usepackage[font=small,labelfont=bf]{caption}
\usepackage[margin=1in]{geometry}
\usepackage{enumitem}
\usepackage{thmtools,thm-restate}
\usepackage{bbm}
\usepackage{nicefrac}
\usepackage[normalem]{ulem}

\setlist{nosep}

\allowdisplaybreaks

\usepackage{mdframed}
\newtheorem*{mdresult}{Result}

\usepackage[usenames,dvipsnames]{xcolor}
\usepackage[backref=page,linktocpage=true,breaklinks,colorlinks,citecolor=blue,linkcolor=BrickRed]{hyperref}

\newcommand{\IGNORE}[1]{}
\newcounter{note}[section]

\usepackage[usenames,dvipsnames]{xcolor}
\usepackage[linktocpage=true,breaklinks,colorlinks,citecolor=blue,linkcolor=BrickRed]{hyperref}
\usepackage{cleveref}
\DeclareUnicodeCharacter{2217}{*}

\newtheorem{Theorem}{Theorem}[section]
\newtheorem{Lemma}[Theorem]{Lemma}

\newtheorem{Definition}[Theorem]{Definition}
\newtheorem{Observation}[Theorem]{Observation}

\newtheorem{Claim}[Theorem]{Claim}
\newtheorem{Corollary}[Theorem]{Corollary}

\newcommand{\parta}{(\text{\uppercase\expandafter{\romannumeral1}})}
\newcommand{\partb}{(\text{\uppercase\expandafter{\romannumeral2}})}
\newcommand{\partc}{(\text{\uppercase\expandafter{\romannumeral3}})}
\newcommand{\bev}{\mathtt{BE}}
\newcommand{\opt}{\mathsf{Opt}}
\newcommand{\esopt}{\widehat{\mathsf{Opt}}}

\newcommand{\obj}{\mathsf{Obj}}
\newcommand{\poly}{\mathsf{poly}}

\newcommand{\alg}{\mathsf{Alg}}

\newcommand{\cons}{a}
\newcommand{\val}{v}

\newcommand{\loss}{\mathsf{Loss \ in \ Revenue}}

\newcommand{\ptt}{\mathcal{P}}

\newcommand{\type}{\mathsf{Type}}

\newcommand{\D}{\mathcal{D}}
\newcommand{\A}{\mathcal{A}}

\newcommand{\E}{\mathbb{E}}
\newcommand{\B}{\mathbf{B}}

\newcommand{\OTild}{\widetilde{O}}

\newcommand{\one}{\mathbf{1}}%
\newcommand{\Ex}[2][]{\mbox{\rm\bf E}_{#1}\left[#2\right]}%
\renewcommand{\Pr}[2][]{\mbox{\rm\bf Pr}_{#1}\left[#2\right]}%

\newcommand{\pr}{\mathbf{P}} 

\newcommand{\ignore}[1]{{}}
\newcommand{\R}{\mathbb{R}}

\newcommand{\calD}{\mathcal{D}}

\newcommand{\bld}[1]{\boldsymbol{#1}}
\newcommand{\price}{\lambda}
\newcommand{\initprice}{\lambda^{\mathtt{INIT}}}

\newcommand{\accprice}{\lambda^*}
\newcommand{\bprice}{\pmb{\lambda}}
\newcommand{\ba}{\bld{a}}
\newcommand{\EA}{\mathtt{EA}}

\newcommand{\calH}{\mathcal{H}}
\newcommand{\consx}{\cons_{i, j}^*{(x)}}
\newcommand{\esconsx}{\widehat{\cons}_{i, j}}

\newcommand{\objx}{\obj{(x)}}

\newcommand{\stoptime}{{\tau}}

\newcommand{\numr}{{n_\mathsf{R}}}
\newcommand{\numg}{{n_\mathsf{G}}}
\newcommand{\lp}{\mathtt{LP}}
\newcommand{\sse}{\subseteq}
\newcommand{\calg}{a^{\mathtt{ALG}}}
\newcommand{\csoln}{a^*}\allowbreak
\newcommand{\tildx}{\tilde x}
\newcommand{\tildz}{\tilde z}
\newcommand{\disceps}{\epsilon}

\title{\textmd{\bf  Single-Sample and Robust 
 Online Resource Allocation}}
\date{\today}
\author{Rohan Ghuge\footnote{Department of Industrial and Systems Engineering / Algorithms and Randomness Center, Georgia Institute of Technology, Atlanta, USA. Email: rghuge3@gatech.edu.} \and  Sahil Singla\footnote{School of Computer Science, Georgia Institute of Technology, Atlanta, GA, USA. Email: ssingla@gatech.edu. Supported in part by NSF awards CCF-2327010 and CCF-2440113.}  \and  Yifan Wang\footnote{School of Computer Science, Georgia Institute of Technology, Atlanta, GA, USA. Email: ywang3782@gatech.edu. Supported in part by NSF awards CCF-2327010 and CCF-2440113.}}

\begin{document}
\maketitle 

\begin{abstract}
\medskip 

Online Resource Allocation problem  is a central problem in many areas of Computer Science, Operations Research, and Economics.
In this problem, we sequentially receive $n$ stochastic requests  for $m$ kinds of shared resources, where each request can be satisfied in multiple ways, consuming different amounts of resources and generating different values. The goal is to achieve a $(1-\epsilon)$-approximation to the hindsight optimum, where  $\epsilon>0$ is a small constant, assuming each resource has a large budget (at least $\widetilde \Omega\big(\poly(1/\epsilon)\big)$). 

\medskip
In this paper, we investigate the learnability and robustness of online resource allocation. Our primary contribution is a novel \emph{Exponential Pricing} algorithm with the following properties:

\vspace{0.1cm}
\begin{enumerate}
    \item It requires only a \emph{single sample} from each of the 
$n$ request distributions to achieve a  $(1-\epsilon)$-approximation for online resource allocation with large budgets. 
    Such an algorithm was previously unknown, even with access to polynomially many samples, as prior work either assumed full distributional knowledge or was limited to i.i.d.\,or random-order arrivals.

\smallskip

\item It is robust to corruptions in the outliers model of \cite{BGSZ-ITCS20} and the value augmentation model of \cite{immorlica2020prophet}. Specifically, it maintains its \((1 - \epsilon)\)-approximation guarantee under both these robustness models, resolving the open question posed in \cite{AGMS-SODA22}. 

\smallskip

\item It operates as a simple item-pricing algorithm that ensures incentive compatibility. 

\end{enumerate}

\medskip

\noindent The intuition behind our Exponential Pricing algorithm is that the price of a resource should adjust exponentially as it is overused or underused. 
It differs from conventional approaches that use an online learning algorithm for item pricing. 
This departure guarantees that the algorithm will never run out of any resource, but  loses the usual no-regret properties of online learning algorithms, necessitating a new analytical approach.

\end{abstract}

\bigskip

\clearpage

\newpage

\clearpage

\clearpage
\section{Introduction}

Online resource allocation is a central problem in  many areas of Computer Science, Operations Research, Economics, and Stochastic Optimization. It captures a commonly occurring tradeoff: should we allocate resources to meet a current request or reserve them for potential future gains? This dilemma underlies numerous applications, including online advertising and matching, online packing linear programs, online routing, and combinatorial auctions.
Formally, we receive $n$ sequential requests for $m$ kinds of shared resources, where each request can be satisfied in multiple ways, consuming different amounts of resources and generating different \emph{values}.
The primary goal is to maximize total value (also called ``welfare") while respecting a predefined \emph{budget} constraint for each resource. 
 The main motivating question in theoretical computer science is: \emph{Given a small constant $\epsilon>0$,   can we design an online algorithm that achieves $(1-\epsilon)$-fraction of the value of an optimal offline algorithm?}

It is well-known that if the $n$ requests are drawn independently from \emph{known distributions}, simple LP rounding  yields  a $(1-\epsilon)$-approximation, provided each resource’s budget is at least  $\widetilde \Omega\big(\poly(1/\epsilon)\big)$\footnote{Notation $\widetilde{\Omega}(\cdot)$ omits $\poly(\log(nm\epsilon^{-1}))$ factors.} and each request consumes no more than one unit of any resource. This \emph{large budget} assumption is justifiable both in practice and in theory: large-scale applications, such as online advertising, often operate with substantial budgets, whereas when budgets are small, an $\Omega(1)$ loss from the optimum value is unavoidable, even in the single-resource case (commonly studied under the prophet inequality).
However, in practical scenarios, input distributions are rarely fully known, especially in high-dimensional settings where the curse of dimensionality poses a significant challenge. 
Consequently, much of the prior research has focused on  the  special case of   i.i.d.\,arrivals, or the closely related random-order arrivals, both under the large budget assumption~\cite{DevenurHayes-EC09,AWY-OR14,MR-MOR14,devanur2019near,KRTV-SICOMP18,agrawal2014fast,GM-MOR16}. 
In these cases, we  discard the first $\epsilon n$  requests, incurring only an  $\epsilon$-fraction loss in the optimum value, and use them to learn a  $\big(1-O(\epsilon)\big)$-approximation online algorithm; see book chapters 
\cite[Chapter 6]{EIV-Book23} and \cite{GS-Book20}.

Unfortunately, these approaches fail in the setting of  \emph{unknown non-identical} distributions, which captures fully adversarial online arrivals, making an $\Omega(1)$-approximation impossible. 
Moreover, existing algorithms for known/identical distributions lack robustness to small corruptions  \cite{BGSZ-ITCS20,immorlica2020prophet}.
Our paper is concerned with two key themes: \emph{single-sample learning} and \emph{robustness}.

\medskip
\noindent\textbf{Single-Sample Learning.} 
A common approach to tackling stochastic problems with unknown distributions is through sample complexity analysis: Can we design $(1-\epsilon)$-approximation algorithms for online resource allocation  with non-identical distributions, 
given only a small number of samples from each of the $n$ distributions? 
Notably, no such algorithm was previously known with finite sample complexity: standard approaches via distribution learning fail for unbounded value distributions. 
At the extreme, one might even hope: \emph{can we design  a $(1-\epsilon)$-approximation algorithm given just a \emph{single sample} from each of the $n$ request distributions?} 
Various online problems have been explored under this single-sample model, such as prophet inequalities~\cite{AzarKW14,RubinsteinWW20}, online matching~\cite{KaplanNR22,CDFFLLP-SODA22}, combinatorial auctions~\cite{FOCS24-toapper}, network design~\cite{GGLS-SODA08}, and load balancing~\cite{ArgueF0S22}.

\medskip
\noindent \textbf{Robustness.} 
In practical settings, some requests may deviate from the distributional assumptions, and 
designing algorithms resilient to such corruptions becomes critical. 
A popular robustness model, rooted in Huber's contamination framework from Robust Statistics \cite{Huber64,DK-Book23}, is that of \emph{outliers}: here $(1-\delta)$ fraction of the requests are stochastic (e.g., i.i.d.), while the remaining $\delta$ fraction of requests are adversarial. 
The benchmark is the value of the stochastic (inlier) requests. In this model, the best known robust  algorithm for online resource allocation gives only an $\Omega(1)$-approximation~\cite{AGMS-SODA22}, except in the special case of a single resource~\cite{BGSZ-ITCS20}. 
The situation is similar in the ``value augmentation'' model of robustness~\cite{immorlica2020prophet,AGMS-SODA22}, where again the best known algorithm gives only an $\Omega(1)$-approximation, except when dealing with a single resource.  
This raises the question: \emph{can we develop  $(1-\epsilon)$-approximation algorithms for online resource allocation with multiple resources that are robust to outliers and augmentations?}

The main contribution of this work is the development of an Exponential Pricing algorithm for online resource allocation with large budgets having the following key properties:
\begin{enumerate}[topsep=0pt,itemsep=-0.5ex,partopsep=1ex,parsep=1ex,leftmargin=0.6cm]
    \item The algorithm requires only limited information about the input distributions, which we show can be estimated from a \emph{single sample} from each of the $n$ distributions. 
    This yields the first single-sample $(1-\epsilon)$-approximation algorithm for non-identical distributions, where  no bounded sample-complexity result was previously known.  Such questions have been highlighted in several recent works, including \cite{FOCS24-toapper,KaplanNR22,CDFFLLP-SODA22,RubinsteinWW20,AzarKW14}.

    \item The algorithm is robust to a variety of corruptions in the requests. Specifically, it maintains  its $(1-\epsilon)$-approximation guarantee in the outliers model of \cite{BGSZ-ITCS20} and the value augmentation model of \cite{immorlica2020prophet}, thereby resolving the main open question posed in \cite{AGMS-SODA22}. 

    \item The algorithm  uses ``posted pricing", i.e., each incoming request faces $m$ resource prices  and takes the greedy action. Such algorithms are desirable for their ease of implementation (e.g., from groceries to airline tickets) and directly imply online truthful mechanisms; e.g., see Lucier's survey \cite{Lucier17}. Consequently, we obtain an online truthful mechanism that achieves a $(1-\epsilon)$-approximation  to maximum  welfare, partially answering the question of \cite{feldman2014combinatorial} on analyzing ``the efficiency of posted price mechanisms as a function of the minimal number of item copies''. 
\end{enumerate}

\noindent 
Before explaining our results in detail,  we begin by presenting a formal framework for online resource allocation to place things in the proper context.

\subsection{Online Resource Allocation} \label{sec:introModel}
We formally define the  online resource allocation problem. We receive a sequence of $n$  requests for $m$ kinds of limited {resources}. Each request $i \in [n]$, denoted $\gamma_i = (\val_i, a_i, \Theta_i)$, can be satisfied in several ways, where choosing  $\theta_i$ in the decision set $\Theta_i$ generates  \emph{value} $v_i(\theta_i) \in \R_{\geq 0}$ and consumes (potentially many) resources $a_i(\theta_i) \in  [0,1]^m$. 
The goal is to  maximize  total value (\emph{welfare}) without exceeding the known  \emph{budget} $B_j$  for any resource $j \in [m]$.  We will assume that there is  a null decision $\phi \in \Theta_i$,  with $\val_i(\phi) = 0$ and $a_i(\phi) = \mathbf{0}$, allowing the  algorithm to ``stop'' after budget exhaustion. 
The  benchmark is the optimal value in hindsight.


If the sequential requests are chosen fully adversarially,  no online algorithm achieves more than $\Omega(B/n)$ fraction of the hindsight optimum, even for the special case of a single resource.
Hence, we  study the problem in the \emph{stochastic model} where each request $\gamma_i$ is independently drawn from some  distribution $\calD_i$. 
    The benchmark is the expected hindsight optimal value:
    \begin{equation}\label{eq:hindsight-opt}
  \textstyle \opt ~~:=~~         \E_{\calD_1, \ldots, \calD_n}\left[ \max_{\{\theta_i \in \Theta_i\}_{i \in [n]}}  \sum_{i=1}^n \val_i(\theta_i) ~\text{ s.t. }  \sum_{i=1}^n a_i(\theta_i) \leq \B  \right].
    \end{equation}

\medskip
\noindent\textbf{Posted Pricing.} In addition to maximizing total  value, we  desire that our algorithms use \emph{posted pricing}: each incoming request $i \in [n]$ faces resource prices $\bprice_i \in \R_{\geq 0}^m$ and takes the \emph{best-response} (greedy) decision  $\arg\!\max_{\theta  \in \Theta_i} \big( v_i(\theta) - \langle\bprice_i, a_i(\theta)\rangle \big)$.
Pricing algorithms are desirable because they are easy to implement and  crucial for game-theoretic applications such as auctions, where self-interested agents might misreport their values \cite{NRTV-Book07}. 
The  challenge is in setting the resource prices: high prices may lead to underutilization, while low prices can cause early depletion of resources.

\noindent Online resource allocation  captures several well-studied problems as special cases. 
\begin{itemize}[topsep=0pt,itemsep=-0.5ex,partopsep=1ex,parsep=1ex,leftmargin=0.6cm] %
    \item \emph{Online Matching and Advertisement.} {In the trillion-dollar online advertisement industry \cite{eMarketer23}, users searching for keywords arrive sequentially and must be immediately shown relevant advertisements.} 
    This is modeled as online bipartite matching, as detailed in Mehta’s book \cite{Mehta-Book13},  a special
    case of online resource allocation with actions determining the  vertex to match to. Other applications of online matching include kidney exchanges, ridesharing, and online dating.

    \item \emph{Online Packing LPs.} We aim to solve the LP  $\{\max c^T x \mid Ax \leq \B \text{ and } x\in [0,1]^n\}$, where columns of $A$ arrive sequentially  and we must immediately set $x_i$. This is captured by online resource allocation when $\Theta_i = [0,1]$ and $\val_i(\theta) = c_i \theta$. Important special cases  include AdWords \cite{MSVV-JACM07}, online routing \cite{AwerbuchAP-FOCS93,BN-FOCS06}, and refugee resettlement \cite{AGPTT-OR24}. 
        
    \item \emph{Online Combinatorial Auctions.} A sequence of $n$ buyers arrive with combinatorial valuations $v_i : 2^{[m]} \rightarrow \R$ (like submodular/subadditive) over $m$ items and we must immediately allocate them a subset of $m$ items. 
    This is a fundamental problem in Algorithmic Game Theory  \cite{feldman2014combinatorial,Alaei-SICOMP14,DFKL-SICOMP20,DKL-FOCS20,CC-STOC23}, and is captured  when   $\Theta_i$ denotes subsets of items.
\end{itemize}

\subsection{Single-Sample Online Resource Allocation}\label{sec:results-single-sample}

Our first main result is that for online resource allocation with large budgets, a single sample from each of the $n$  distributions is sufficient to achieve a $(1-\epsilon)$-approximation.

\begin{Theorem}
\label{thm:intro-main}
Given a single sample from each of the $n$ request distributions, there exists a $(1-\epsilon)$-approximation algorithm for online resource allocation with \emph{non-identical} distributions, provided that each resource has ${\Omega}\big(\poly({1}/{\epsilon})\cdot \poly(\log(nm\epsilon^{-1}))\big)$ budget.  Moreover, this is a posted pricing algorithm,  
thereby immediately enabling a $(1-\epsilon)$-approximation  online truthful mechanism  to the maximum  welfare.
\end{Theorem}

Notably, no $(1 - \epsilon)$-approximation algorithm for  online resource allocation was previously known with any finite number of samples (even without the posted pricing requirement). 
For known input distributions, however, it is folklore that randomized rounding of the standard LP (see \Cref{sec:LPandGoodEstimates}) already achieves a $(1-\epsilon)$-approximation when the budgets are ${\Omega}({\log (m)}/{\epsilon}^2)$.  
This LP rounding approach does not extend to unknown distributions because the support size could be infinite and the values  unbounded, making distribution learning impossible  with a finite number of samples.

Moreover, we show that $\Omega\big({\log(m)}/{\epsilon^{2}}\big)$ budget for every resource is necessary to obtain a $(1 - \epsilon)$-approximation, even for known input distributions. This proof in \Cref{sec:hardnessExample} is based on a  modification of the hardness example of  \cite{AWY-OR14} in the related secretary model.

Before describing our techniques for proving \Cref{thm:intro-main}, we discuss the limitations of prior works that focused on specific value functions or identical distributions.

\medskip
\noindent \textbf{Limitations of Prior Approaches.} 
Prior to our work, single-sample $\Omega(1)$-approximation algorithms for online resource allocation have been developed for certain value function classes \cite{AzarKW14,KaplanNR22,CDFFLLP-SODA22,FOCS24-toapper}. In particular, the recent paper \cite{FOCS24-toapper} obtains $\Omega(1)$-approximation for any submodular/XOS value functions. However, this line of work is in the ``small-budget''  setting with budget $B_j=1$ for each resource, where $\Omega(1)$-approximation   is impossible for general value functions \cite{DFKL-SICOMP20,CCFPW-MP23}. Consequently, their techniques must exploit special properties of the value function, and are inapplicable to the large-budget setting with general value functions.

For \emph{identical} distributions, or the related random-order model, a long line of work has shown that a   $(1 - \epsilon)$-approximation can be achieved using dynamic posted prices, provided the budgets are $\widetilde{\Omega}(\poly({1}/{\epsilon}))$
\cite{DevenurHayes-EC09,AWY-OR14,MR-MOR14,devanur2019near,agrawal2014fast,GM-MOR16}\footnote{These works don't require any samples  since for identical distributions we can interpret a small number of the initial requests as samples.}. 
Their main idea is to think of the price of $m$ resources as a distribution over $m$ experts, and then update these prices using an online learning algorithm like Multiplicative Weights or Hedge \cite{AHK-ToC12}. The $(1 - \epsilon)$-approximation analysis  crucially relies on the ``no-regret'' property of online learning.

Two major challenges arise in extending dynamic pricing algorithms to non-identical distributions. 
First, a dynamic pricing algorithm for the non-identical distributions setting is not known in the literature, even if all the $n$ request distributions were known. Existing methods for i.i.d.\,(or random-order) arrivals rely on consuming an average budget $\B/n$ per request, which doesn’t hold for non-identical distributions. 
{Second, when the distributions are unknown and we only have a single sample per distribution, the algorithm must only rely on distributional properties  that can be estimated from a single sample.}

 Next, we outline our approach to addressing these challenges through a new Exponential Pricing algorithm,   and then we discuss how to implement it with a single sample.

\medskip
\noindent \textbf{Exponential Pricing Algorithm.} Our new algorithm requires as input, for each time step $i$, 
the expected budget consumption up to that point, 
$\sum_{\ell \, < \, i} \E[a_{\ell}(\theta^*_\ell)] $, 
based on optimal decisions, $\theta^*_\ell$. Additionally,
it requires the expected optimal value $\opt$, as defined in \eqref{eq:hindsight-opt}. 
When request distributions are known, these inputs are straightforward to compute; later, we show how to estimate them from a single sample. 
Now, the  \emph{Exponential Pricing} algorithm sets resource prices $\bprice_i \in \R_{\geq 0}^m$ for the $i$-th request as:
\vspace{-0.2cm}
\begin{align} \label{eq:expPricingHighLevel}
 \textstyle
  \bprice_{i} = \initprice  \cdot \exp\Big(\delta \cdot \left(\,\sum_{\ell \, < \, i} a_{\ell}(\theta_\ell) - \sum_{\ell \, < \, i} \E[a_{\ell}(\theta^*_\ell)]  \right)\Big), 
\end{align}
where the parameters are set as $\initprice  \approx \frac{\opt}{nm}$ and $\delta \approx \frac{1}{\epsilon \cdot B}$, and $\exp$ is applied coordinate-wise. 
This means that the price of each resource increases exponentially based on the difference between its actual consumption and its expected consumption.

The Exponential Pricing algorithm is reminiscent of the famous online learning algorithm, Hedge, due to Freund and Schapire \cite{FS-JCSS97}.   The major difference is that 
\eqref{eq:expPricingHighLevel} is an equality,  not proportionality,
meaning we do not renormalize to maintain a distribution over $m$ experts, allowing the sum of resource prices to increase exponentially. 
This modification is essential in our analysis: it allows us to prove that, with high probability,  Exponential Pricing  never exhausts any resource budget.
This is because the price of a resource increases exponentially if it is overused, meaning that a resource price would exceed $\opt$ before exhaustion,  leading to a contradiction. 
However, this change from an online learning algorithm to Exponential Pricing sacrifices the  ``no-regret'' property that previous algorithms used for identical distributions.

\medskip
\noindent \textbf{No-Regret Property.} To introduce the no-regret concept, consider a standard online learning setup: at each time step $i$, the algorithm first chooses a distribution  $\mathbf{p}_i$ in the $m$-dimensional full simplex $\blacktriangle_m$, and then receives a reward $ \langle \mathbf{p}_i , \mathbf{r}_i \rangle$, where $\mathbf{r}_i \in [-1,1]^m$ is a reward vector that is revealed after the  step.

An online learning algorithm is said to have \emph{no regret} w.r.t.\,$\mathbf{p} \in \blacktriangle_m$ if
\[ \textstyle
         \sum_{i = 1}^n \langle \mathbf{p}_i , \mathbf{r}_i \rangle  \geq   \sum_{i = 1}^n \langle \mathbf{p} , \mathbf{r}_i \rangle -  o(n),
        \] 
meaning that the algorithm's cumulative reward is at least that of fixed distribution $\mathbf{p}$, up to a sublinear regret term. Hence, the average regret over $n$ steps goes to $0$ as $n\rightarrow  \infty$.

In previous analyses of online learning based  algorithms for i.i.d.\,arrivals \cite[Chapter 6]{EIV-Book23}, two cases are typically considered:
{either some resource $j$ exhausts its budget $B_j$, accumulating sufficient value, or no resource is exhausted by the end, and the algorithm should have played the null price in hindsight. In the first case, we apply the no-regret property  w.r.t.\,$\mathbf{p} = e_j$, and in the second case w.r.t.\,$\mathbf{p} = \mathbf{0}$.}
A key insight in our analysis is that, while  Exponential Pricing  does not satisfy the no-regret property relative to every fixed price vector (since it does not even play a distribution), it does satisfy the no-regret property w.r.t.\,the $\mathbf{0}$ price vector. This specialized no-regret property,  together with the fact that Exponential Pricing never exhausts any resource’s budget with high probability, enables us to achieve  a  $(1-\epsilon)$-approximation for non-identical distributions.

\medskip
\noindent \textbf{Exponential Pricing with a Single Sample.}  
To use Exponential Pricing with unknown distributions, we need to estimate 
$\opt$ and the expected prefix budget consumptions,  $\sum_{\ell \, < \, i} \E[a_{\ell}(\theta^*_\ell)]$, from the single sample. Our approach involves two main steps:
\begin{enumerate}[topsep=0pt,itemsep=-0.5ex,partopsep=1ex,parsep=1ex,leftmargin=0.6cm]
    \item \emph{Tolerance to Estimation Errors}: We show that $\poly(mn)$-factor  errors in the estimate of  $\opt$ are tolerable, as they only logarithmically affect the initial price  $\initprice$, yielding negligible impact on budget constraints in \Cref{thm:intro-main}.  
    For prefix budget consumption, we prove that our analysis of  Exponential Pricing can accommodate additive errors of  $O(\epsilon^2 B)$.  Roughly, this is because we can absorb this error in the ``standard-deviation'' errors.

    \item \emph{Single-Sample Estimation}: Estimating $\opt$ to within a $\poly(mn)$ factor is straightforward using the highest $O(1/\epsilon)$ value requests. The main difficulty lies in estimating prefix budget consumptions accurately to within $O(\epsilon^2 B)$. Observe that such an estimate is easy with $\OTild(\epsilon^{-4})$ independent samples as they provide unbiased prefix consumption estimates. 
    Since we only have a single sample, a natural idea is to uniformly partition the sample into $\OTild(\epsilon^{-4})$ parts, treating each part as independent.    
    Although this introduces correlations, we mitigate this by defining a ``local" optimal solution for each part as  our learning target, rather than the global optimal solution. This approach bypasses correlations, allowing us to apply concentration bounds. Additionally, we show that when the partition is chosen uniformly at random, the local solutions combine to a near-optimal global  solution, thereby still yielding a $\big(1-O(\epsilon)\big)$-approximation. 
\end{enumerate}

\subsection{Robust Online Resource Allocation}\label{sec:results-techniques}

Another major focus of this work is on showing that Exponential Pricing is a \emph{robust} online resource allocation algorithm. As mentioned earlier, the motivation for robustness comes from practical settings where some of the requests do not satisfy the stochastic modeling assumptions, like due to noise or a few agents acting adversarially.   Hence, we would like to design algorithms whose performance does not (significantly) degrade due to   ``outliers'' or ``augmentations''.

\medskip
\noindent\textbf{Outliers.} Inspired by the classic Huber's contamination model from  Robust Statistics \cite{Huber64,DK-Book23}, the authors of \cite{BGSZ-ITCS20} introduce the Byzantine Secretary model for robust online resource allocation. Here, $(1-\delta)$ fraction of the requests are drawn  stochastically (arriving in a random order) and the remaining $\delta$ fraction of the requests are \emph{outliers} that are chosen adversarially (both the request and time of arrival). The goal is to obtain a $(1-O(\epsilon))$-approximation to the stochastic/inliers part of the input. Similar models have  been also explored for special cases of online resource allocation in \cite{KM-ICALP20} and \cite{GKRS-ICALP20}. 

The main result of \cite{BGSZ-ITCS20} is a $(1-\epsilon)$-approximation algorithm, in the special case of a \emph{single resource}, that is robust to outliers. Interestingly, this guarantee is independent of the fraction of outliers $\delta$.
 In \cite{AGMS-SODA22}, the authors study robust online resource allocation for the general case of $m$ resources. They give an $\Omega(1)$-approximation algorithm by observing that online learning based pricing algorithms are robust to outliers. They also give a different $(1-\delta - \epsilon)$-approximation algorithm when at most $\delta$ fraction of the inputs are outliers. 
Unfortunately, as discussed later, their techniques do not extend to give a $(1-\epsilon)$-approximation. 

\medskip
\noindent\textbf{Augmentations.} 
For online resource allocation with non-identical distributions, 
the authors of \cite{immorlica2020prophet} introduce an ``augmentation'' model of robustness, where an adversary is  allowed to arbitrarily \emph{augment} (increase) the value of the requests (but not decrease). The main question is whether this  degrades the algorithm's performance. It was observed in  \cite{immorlica2020prophet}  that the performance of  popular online algorithms degrades significantly due to augmentations, but they could design a $(1-\epsilon)$-approximation robust algorithm in the special case of a \emph{single resource}.  The authors of \cite{AGMS-SODA22} could again apply their techniques to obtain an $\Omega(1)$-approximation for the general case of $m$ resources with augmentations, but obtaining  a $(1-\epsilon)$-approximation has remained open.

At a high-level, the    limitation of the approach in \cite{AGMS-SODA22} is that their analysis has two cases: either the algorithm stops early since it runs out of budget for a resource, or the algorithm goes till the end. In the former case they argue that the ``revenue'' is large since we exhausted a resource, and in the latter case they argue that the ``utility'' is large since we always have all the resources. This approach loses at least a factor of $1/2$ to balance between these two cases.

Our main result gives a $(1-\epsilon)$-approximation pricing algorithm for online resource allocation that is robust to both outliers and augmentations.

\begin{Theorem}[Informal version of Theorems~\ref{thm:romainori} and \ref{thm:onlineRA-aug}]
\label{thm:intro-robust}
Given an $\epsilon > 0$, there exists a $(1-\epsilon)$-approximation algorithm for  online resource allocation that is robust to both outliers and augmentations, provided that the  budget for every resource is $\tilde{\Omega}({1}/{\epsilon^2})$. 
\end{Theorem}

 We prove \Cref{thm:intro-robust} by showing the Exponential Pricing is the desired algorithm. For robustness w.r.t.\,augmentations, we show that the previous analysis still works as augmentations can only make our algorithm run out of budget early, but that is the easy case for  Exponential Pricing  since then the  prices are even higher than the optimal value.

 Proving robustness w.r.t.\,outliers in the Byzantine Secretary model requires more work. The idea  is to still apply  Exponential Pricing. 
  Although the proof structure resembles previous analysis, a  key difference arises due to a dependency in the current price vector and the current request. This is because when requests arrive   in random order, the current request depends on previous requests, creating a correlation with the current price vector. To address this, we interpret the randomness in historical data as a ``sampling without replacement” process and apply concentration inequalities for this type of sampling to control the added correlation. Performing this effectively requires some other modifications like restarting the algorithm at the midpoint.

\subsection{Further Related Work}

\noindent \textbf{Prophet and Secretary models.} 
Originating in Optimal Stopping Theory approximately 50 years ago \cite{Dynkin, KrengelS77, KrengelS78}, prophet and secretary problems have become central in TCS over the last two decades. These models are especially valuable in TCS for two reasons: (1)~they enable ``beyond the worst-case" analysis of online problems, where good algorithms are otherwise impossible under adversarial arrivals, and (2)~they are an important tool in Algorithmic Mechanism Design, often leading to near-optimal posted-pricing mechanisms. For a comprehensive overview, we refer  to the survey~\cite{Lucier17} and book chapters \cite{GS-Book20} and \cite[Chapter 30]{EIV-Book23}.

The most relevant part of this literature to our work is on  online packing LPs for i.i.d.\,or random-order arrivals. This line of study began with the work of  Devanur and Hayes \cite{DevenurHayes-EC09} for matchings/AdWords, and then extended to general packing LPs in \cite{AWY-OR14,FHKMS-ESA10,MR-MOR14}. These initial works showed that $(1-\epsilon)$-approximation is possible if the budget is at least $\poly(m/\epsilon)$. The dependency on the number of resources $m$ was exponentially improved in subsequent works to $\poly(\log m/\epsilon)$, initially through non-pricing algorithms \cite{KRTV-SICOMP18}, and later through pricing algorithms  \cite{agrawal2014fast,GM-MOR16,devanur2019near,BLM-OR23}.   A recent work \cite{BHKKO-SODA24} shows how to convert online algorithms into pricing algorithms, but they price  (exponentially many) 
 subsets of resources, instead of  individual resources, and are inapplicable with only a single sample per distribution. 

Additionally, online packing LPs have been explored in the context of regret minimization, where the benchmark is the best fixed distribution over decisions \cite{BKS-JACM18, AD-OR19, ISSS-JACM22}.  These  techniques  are   similar in spirit to the above techniques for obtaining a  $(1-\epsilon)$-approximation ratio.

\medskip
\noindent\textbf{Sample complexity.}
Motivated by developments in Machine Learning,   data-driven algorithms  seek to learn the optimal algorithm for inputs arriving from an unknown distribution; for an overview, see  book chapter \cite{Balcan-Chp20}.
In the context of prophet problems, the sample complexity of various packing problems has been investigated. It is known, e.g., that a single sample is sufficient to achieve $\Omega(1)$-competitive   algorithms for matchings and restricted classes of matroids \cite{AzarKW14,RubinsteinWW20,KaplanNR22,CDFFLLP-SODA22,correa2022two}. However, obtaining the optimal competitive ratio, up to an additive $\epsilon$, requires $\poly(1/\epsilon)$-samples, even for single-item prophet problems \cite{CorreaDFS19,RubinsteinWW20,CristiZilliotto24,correa2024sample}.

\medskip
\noindent\textbf{Adversarial arrivals or small budgets.} 
For adversarial arrivals, achieving a bounded competitive ratio is generally impossible in online resource allocation. However, assuming bounded values, logarithmic competitive ratios can be obtained when either fractional decisions are allowed or there is budget augmentation \cite{AwerbuchAP-FOCS93, BN-FOCS06, BN-MOR09, BuchbinderNaor-Book09, AzarBCCCG0KNNP16}. In the AdWords setting, 
$(1-1/e)$-competitive algorithms are achievable under a large-budget assumption \cite{MSVV-JACM07, DJ-STOC12}.

With small budgets, achieving a $(1-\epsilon)$-competitive-ratio is infeasible in either of prophet or secretary models. However, an $\Omega(1)$ competitive-ratio is possible for certain submodular/XOS valuations, both in the prophet model \cite{feldman2014combinatorial,DKL-FOCS20,CC-STOC23} and in the secretary model \cite{KRTV-ESA13,FOCS24-toapper}. A major open question is to obtain an $\Omega(1)$ competitive-ratio using posted pricing  for submodular/XOS valuations in the secretary model \cite{AS-FOCS19,AKS-SODA21}.
For general valuations,  $\Theta(1/s)$ competitive-ratio is achievable when each request requires at most $s$ resources \cite{CCFPW-MP23,DFKL-SICOMP20,MRST-OR20,MMZ-arXiv24}.

\subsection{Paper Organization}

In \Cref{sec:exp-pricing-estimates}, we provide our  Exponential Pricing algorithm for stochastic Online Resource Allocation, assuming  ``good estimates'' of certain parameters of the $n$ request distributions  are provided. 
In \Cref{sec:single-sample}, we show how to obtain these good estimates using a single sample from each of the $n$ distributions, and prove \Cref{thm:intro-main}. 

We then apply the Exponential Pricing algorithm to the Byzantine Secretary model and the Prophet with Augmentations model in \Cref{sec:robustness} and \Cref{sec:sub-augmentation}, respectively, and show that it is robust against both outliers and augmentations.


\newcommand{\supx}{(x)}

\section{Exponential Pricing given Good Estimates}
\label{sec:exp-pricing-estimates}

In this section, we give and analyze the Exponential Pricing algorithm in the Prophet model. 
To recap, here $n$ requests arrive sequentially, where the $i$-th request $\gamma_i = (\val_i,a_i, \Theta_i)$ is independently drawn from a distribution $\calD_i$.
Notably, we show that the Exponential Pricing algorithm
requires only specific resource consumption estimates
from the distributions $\calD_1, \ldots, \calD_n$ to obtain a nearly-optimal solution.
This will be useful, as we show in the following \Cref{sec:single-sample} that these estimates can be learned using a single sample from each distribution, thus proving  \Cref{thm:intro-main}. Before we state the main result of this section, we will first outline some useful preliminaries and definitions.

\subsection{Linear Program and Good Estimates} \label{sec:LPandGoodEstimates}

To define the required estimates, we first present a linear program (a.k.a. configuration LP) that provides an upper bound on the expected hindsight optimum (see~\eqref{eq:hindsight-opt}). 
For concreteness, we assume that there are $K$ different request tuples (types) $\gamma_{i,1}, \ldots, \gamma_{i,K}$ for the $i^{\text{th}}$ request, which equals $\gamma_{i,k} = (\val_{i,k}, a_{i,k}, \Theta_{i,k})$ with probability $p_{i, k}$, where  $\sum_{k=1}^K p_{i, k} = 1$.
We define variables $x_{i, k, \theta}$ 
to denote the fractional allocation to  decision $\theta \in \Theta_{i, k}$ for $\gamma_{i,k}$ being the $i^{\text{th}}$ request. The LP says:
\begin{align}
    \text{maximize } \quad \quad  &  \textstyle  \quad\sum_{i \in [n]}  \sum_{k=1}^{K} \sum_{\theta \in \Theta_{i, k}} \val_{i,k}(\theta) \cdot x_{i, k, \theta},\notag \\
    \text{s.t.}\qquad\qquad & \textstyle \quad   \sum_{i \in [n]} \sum_{k=1}^{K}\sum_{\theta \in \Theta_{i, k}}  a_{i,k}(\theta) \cdot x_{i, k, \theta}  \leq \B \tag{$\text{LP}_{\textsc{UB}}$} \label{program:ora-ub}  \\
    \forall i \in [n], k \in [K] &  \textstyle \quad \sum_{\theta \in \Theta_{i,k}} x_{i, k, \theta} \leq p_{i, k} \notag \\
    \forall i \in [n],  k \in [K], \theta \in \Theta_{i, k},&  \quad  0 \leq x_{i, k, \theta} \leq 1. \notag
\end{align}

        { 
        For explicitly given $n$  request distributions we can solve this LP in $\poly(n,m,K)$ time, independent of the size of the decision sets $\{\Theta_i\}_i$, assuming  \emph{demand oracle}\footnote{For any prices  $\bprice \in \R_{\geq 0}^m$,  the demand oracle for request 
        $\gamma_i = (\val_i, a_i, \Theta_i)$
        returns  $\arg\!\max_{\theta  \in \Theta_i} \big( v_i(\theta) - \langle\bprice, a_i(\theta)\rangle \big)$ .} access to the requests.   This is a standard technique in combinatorial auctions, e.g., see \cite[Chapter 11]{NRTV-Book07} for details. 
        In general, however, the number of request tuples $K$ could be arbitrarily large. Fortunately, our single-sample algorithm will not need to solve \ref{program:ora-ub}, and will only rely on the  observation that this LP provides an upper bound on the optimum.
        }

\begin{Observation}\label{lem:hindsight-opt-ub}
Given an instance of online resource allocation in the prophet model, \ref{program:ora-ub} gives an upper bound to the value of the expected hindsight optimum, $\opt$, as defined in \eqref{eq:hindsight-opt}.
\end{Observation}

We now define the following quantities for any feasible solution $x = \{x_{i, k, \theta}\}$ to \ref{program:ora-ub}:
\begin{itemize}
    \item   $\objx :=~ \sum_{i, k, \theta} \val_{i, k}(\theta)\cdot x_{i,k,\theta}$, which denotes the objective value of $x$.
    \item  $ \cons_{i, j}^*{\supx} :=~ \sum_{k = 1}^K \sum_{\theta \in \Theta_{i, k}} \left(a_{i, k}(\theta)\right)_j \cdot x_{i, k, \theta}$, which denotes the amount of resource $j$ consumed by request $i$ under solution $x$, for all $i \in [n]$ and $j \in [m]$.  
\end{itemize}
We use $\lp\left(\B; \{\calD_i\}_{i \in S}\right)$ to denote the corresponding linear program for a subset $S \sse [n]$ of the requests and the corresponding budget vector $\B$ (obtained by replacing $i \in [n]$ in \ref{program:ora-ub} with $i \in S$).
Consequently, \ref{program:ora-ub} is equivalent to $\lp\left(\B; \{\calD_1, \ldots, \calD_n\}\right)$.
When $S = [n]$, we will drop the argument $\{\calD_1, \ldots, \calD_n\}$ from $\lp(\cdot; \cdot)$. 

The main result of this section shows that, given ``good'' estimates of $\opt$ and $\big\{ \cons^*_{i, j}{\supx} \big\}$, there exists a pricing algorithm that obtains a value  nearly equal to $\objx$, while respecting the budget constraints.
The following definition formalizes what we mean by good estimates.

\begin{Definition}[Good Estimates]
\label{def:goodestimates}
    We say that $\esopt$ and   $\left\{\widehat{\cons}_{i, j}\right\}$   are \emph{good estimates} of $\opt$ and  $\big\{\cons^*_{i, j}{\supx}\big\}$ 
    with a  known parameter   $\beta \geq 1$  
    if     the following holds:
    \begin{itemize}
        
        \item for all $i, j$, we have
        $ 0 \,\, \leq \,\, \widehat{\cons}_{i, j} \,\, \leq \,\,  1 \quad \text{and} \quad 
        \big|\sum_{\ell \leq i}\widehat \cons_{\ell, j} - \sum_{\ell\leq i}\cons_{\ell, j}^*{\supx}\big| \,\, \leq \,\, \frac{\epsilon^2 \cdot B_j}{16 \log (nm\beta/\epsilon)}.$

        \item $\esopt \leq \opt$ and  $\pr_{(\gamma_1, \cdots, \gamma_n)}\left[\sum_{i \in [n]} \one \big[\max_{\theta \in \Theta_i} \val_i(\theta) \geq \beta \cdot \esopt\big] > \frac{10}{\epsilon}\right] \leq \epsilon$, which means that  the number of high-value requests  with respect to $\esopt$ is $O(1/\epsilon)$. 
    \end{itemize}
\end{Definition}

%
%
The last condition above is natural as it is implied by the stronger condition requiring $\opt \leq \beta\cdot \esopt$. Given good estimates, the following is the main result of this section.
\begin{Theorem}\label{thm:exp-pricing-estimates}
    Let $\epsilon > 0$ be an error parameter, and let $x = \{x_{i, k, \theta}\}$ be a feasible solution to $\lp\left(\B\cdot(1-\epsilon)\right)$.
    Furthermore, suppose that we are given \emph{good estimates} $\esopt$ and $\big\{\widehat{\cons}_{i, j}\big\}$ of $\opt$ and $\big\{\cons_{i, j}^*{\supx}\big\}$ with parameter $\beta \geq 1$.   
    Then, there exists a posted pricing algorithm  for online resource allocation in the prophet model that obtains
    expected total value at least $(1-\epsilon)\cdot \objx - 7\epsilon \cdot \opt$, provided that $B_{j} = \Omega\left(\log(nm \beta /\epsilon)/ \epsilon^{2} \right)$ for every resource $j \in [m]$. 
\end{Theorem}

On setting $x = \{x_{i, k, \theta}\}$ to the optimal solution to $\lp\left(\B\cdot(1-\epsilon)\right)$,  we obtain the following guarantee on the performance of Exponential Pricing for online resource allocation in the prophet model with known distributions.

\begin{Corollary}\label{cor:exp-pricing}
    Given an error parameter $\epsilon >0$, there exists a posted pricing algorithm  for online resource allocation in the stochastic model with $n$ non-identical distributions that obtains
    expected total value at least $\big(1 - O(\epsilon)\big)$ times the expected hindsight optimum, provided that $B_{j} = \Omega\left({\log(mn/\epsilon)}\cdot{\epsilon^{-2}}\right)$ for every resource $j \in [m]$.
\end{Corollary}

We note two things about this result. 
First,  this is the first \emph{posted pricing} algorithm that gets a $(1 - \epsilon)$-approximation for online resource allocation with non-identical distributions when the budgets are large. 
Earlier pricing based results either lose a constant factor in the approximation guarantee~\cite{feldman2014combinatorial, DFKL-SICOMP20}, or require the additional assumption that the request distributions are identical~\cite{DevenurHayes-EC09, AWY-OR14, devanur2019near}.
Second, by setting parameters appropriately, we can slightly improve the budget bound in \Cref{cor:exp-pricing} to obtain  $(1 - O(\epsilon))$-approximation as long as  $B_{j} = \Omega\big({\log(m/\epsilon)}/{\epsilon^{2}}\big)$ for every $j \in [m]$. 
This is nearly the best-possible budget bound achievable due to our $\Omega\big({\log(m)}/{\epsilon^{2}}\big)$ budget lower bound in \Cref{sec:hardnessExample}.

In the remainder of this section, we state our Exponential Pricing algorithm and prove \Cref{thm:exp-pricing-estimates}. 
We will assume that  $\epsilon \in (0,1/2]$ is fixed, since \Cref{thm:exp-pricing-estimates} is trivially true when $\epsilon \geq 1/2$.

\subsection{Exponential Pricing Algorithm} The algorithm uses the estimates on the cumulative consumption (with respect to some unknown solution $\{x_{i, k, \theta}\}$)
and its past decisions to set prices on each resource.
More concretely, the price of a resource before the $i$-th request arrives is directly proportional to the difference between the algorithm's cumulative allocation and the estimated cumulative allocation for that resource for the first $i-1$ requests (see \eqref{eq:learning-prices}).

Given the prices, the algorithm makes best-response decisions, i.e, the $i$-th decision $\theta_i$ equals $\arg\!\max_{\theta \in \Theta_i} \left( \val_i(\theta) - \langle\bprice_i, a_i(\theta)\rangle \right)$ where $(\val_i, a_i, \Theta_i)$ denotes the $i^{\text{th}}$ request and $\bprice_i$ denotes the prices before the $i^{\text{th}}$ arrival. 
We use $\calg_{i, j}$ to denote the consumption of resource $j$ by the $i$-th request; so, $ (\calg_{i,1}, \cdots, \calg_{i,m}) =  a_i(\theta_i)$.
Lastly, the algorithm terminates in the $i$-th step if for any resource $j \in [m]$ the budget consumption $ \sum_{\ell \leq i} \calg_{\ell, j} \,\, \geq \,\,  \sum_{\ell \leq i}   \widehat{\cons}_{\ell, j} + \frac{1}{2}\epsilon\cdot B_j$, i.e.,  resource $j$  is   consumed significantly over its estimated allocation. See \Cref{alg:prsp} for a formal description.


\begin{algorithm}
\caption{\textsc{Exponential Pricing via Estimates}}
\label{alg:prsp}
\begin{algorithmic}[1]
\State \textbf{input:} instance $\mathcal{I}$, budgets $\{B_j\}_{j \in [m]}$, error parameter $\epsilon$, estimates $\esopt$, $\left\{\widehat \cons_{i, j}\right\}$
\For{$i = 1, \ldots n$}
\State $(\val_i, a_i, \Theta_i) \gets $ request $i$
\State $\bprice_i \gets $ vector denoting item prices before request $i$ arrives;
for $j \in [m]$, we set %
\begin{equation}\label{eq:learning-prices}
\textstyle  \price_{i,j} \gets \initprice \cdot \exp\left(\delta \cdot \sum_{\ell \, < \, i} \left( \calg_{\ell, j} - \widehat \cons_{\ell, j}\right)\right)
\end{equation}
where $\initprice \gets \esopt \cdot \frac{4\log (nm\beta/\epsilon)}{nm}$ and $\delta \gets \frac{8 \log(nm\beta/\epsilon)}{\epsilon \cdot B_j}$.
\State  $\theta_i = \arg\max_{\theta \in \Theta_i} \left( \val_i(\theta) - \langle\bprice_i, a_i(\theta) \rangle \right)$, and the algorithm gains $\val_i(\theta_i)$.
\State Set vector $ (\calg_{i,1}, \cdots, \calg_{i,m}) =  a_i(\theta_i)$.
\If{there exists resource $j \in [m]$ such that $ \sum_{\ell \leq i} \calg_{\ell, j} \,\, \geq \,\, \sum_{\ell \leq i}\widehat \cons_{\ell, j} + \frac{1}{2}\epsilon\cdot B_j$}
\State {\bf terminate}
\EndIf
\EndFor
\end{algorithmic}
\end{algorithm}

\subsection{Analyzing the Algorithm}
We analyze the total value obtained by \Cref{alg:prsp} and prove \Cref{thm:exp-pricing-estimates}. 
Let $\alg$ denote the random total value obtained by our algorithm. To prove \Cref{thm:exp-pricing-estimates}, we need to show that $\E[\alg] \geq (1-\epsilon)\cdot \objx - 7\epsilon \cdot \opt$, when we are given good estimates of $\opt$ and $\big\{ \cons^*_{i, j}{\supx} \big\}$.

\medskip
\noindent {\bf Proof Overview.} To bound the value obtained by our algorithm, we define an event $\EA$ (short for ``excess allocation'')
to capture its termination condition, which roughly  captures the situation when the budget for one of the resources gets exhausted.
We show (in \Cref{lemma:prob-ea-small}) that the probability of event $\EA$ is negligible, and hence the analysis reduces to bounding the value obtained in the case when all resources are available; that is, when $\EA$ does not occur.

The analysis then proceeds by applying the best-response property to show (in \eqref{eq:for-augmentation}) that
    \[
    \textstyle\E[\alg] ~\geq~  \E \big[ \sum_{i = 1}^\stoptime \val_i(\theta^*_i) \big]  - {\E\big[\sum_{i = 1}^{\stoptime} \langle \bprice_i,   a_i(\theta^*_i) - a_i(\theta_i)  \rangle \big]},
    \]
    where $\stoptime$ denotes the random stopping time and  $\{\theta^*_i\}_{i \in [n]}$ denote  decisions taken according to the solution $\{x_{i, k, \theta}\}$.
    We first show that the first term $\E[ \sum_{i = 1}^\stoptime \val_i(\theta^*_i)] \approx (1-\epsilon)\cdot \objx$.  
    Next, we show that the second term  can be interpreted as ``loss in revenue'' of the algorithm and is $\approx 0$. 
    The heart of its proof is the idea  that  
    the   Exponential Pricing algorithm (even when run with estimates) has a no-regret property w.r.t.\,the $\mathbf{0}$ price vector.

\begin{proof}[Proof of \Cref{thm:exp-pricing-estimates}.]
We first verify that \Cref{alg:prsp} is feasible; i.e., it always returns a solution that satisfies the budget constraints.
We note that \Cref{alg:prsp} continues for request $i$ only when $\sum_{\ell < i} \calg_{\ell, j} \,\, <  \,\,  \sum_{\ell < i}\widehat\cons_{\ell, j} + \frac{1}{2}\epsilon B_j$ holds for all $j \in [m]$. 
Thus, the algorithm allocates at most 
\begin{align*}  
        \sum_{i=1}^n \calg_{i, j} ~<~  \sum_{i=1}^{n-1}\widehat\cons_{i,j} + \frac{1}{2}\epsilon \cdot B_j + \calg_{n, j}
        ~\leq~ \sum_{i=1}^n\cons_{i, j} + \frac{\epsilon^2 B_j}{16 \log (nm\beta/\epsilon)} + \frac{1}{2}\epsilon \cdot B_j + 1  \,\, \leq \,\,  B_j,
\end{align*}
for any resource $j \in [m]$. 
Above, the second inequality follows from the assumptions on the estimates $\widehat \cons_{i, j}$ and that the maximum possible allocation is $1$, and the final inequality follows from the fact that $x$ is a feasible solution to $\lp(\B\cdot(1-\epsilon))$ and $B_{j} = \Omega\left(\log(nm \beta /\epsilon)\cdot \epsilon^{-2} \right)$.

Next, we need to show that \Cref{alg:prsp} obtains good total value. As mentioned earlier, we do this by 
defining an event $\EA$ (short for ``excess allocation'')
which corresponds to
Algorithm~\ref{alg:prsp} terminating for some $i \in [n]$ and $j \in [m]$; that is,
after some request $i$ makes its selection, the allocation for some resource $j \in [m]$ exceeds the estimated budget consumption by $\frac12 \epsilon B_j$.
The following lemma 
shows that we can upper bound the probability of event $\EA$,
which in turn implies that our algorithm does not trigger the termination condition.
In particular, we prove the following.

\begin{Lemma}
\label{lemma:prob-ea-small}
    $\pr(\EA) \leq \epsilon$.
\end{Lemma}

The proof of \Cref{lemma:prob-ea-small} proceeds by showing that the total value obtained from the final unit of the resource that triggered the event is high. This is a crucial feature of our algorithm, which allows the resource price to increase exponentially. 
As a result, if $\pr(\EA)$ is high, the algorithm would obtain a value greater than $\opt$, leading to a contradiction.
We defer the proof of \Cref{lemma:prob-ea-small} to \Cref{sec:NoRegretAndMissingProofs}.

The crux of the proof of \Cref{thm:exp-pricing-estimates}  lies in showing that our algorithm obtains good total value given that $\pr(\EA) \leq \epsilon$. 
Crucially, in this case, the algorithm reaches the end with probability at least $1 - \epsilon$. 
Exploiting this fact, we show the algorithm achieves total value $(1-\epsilon)\cdot \objx - 7\epsilon\cdot \opt$.

To begin, let $\stoptime$ be a random variable denoting
the time at which the algorithm stops.
We note that $\pr(\stoptime = n) = 1 - \pr(\EA) \geq 1 - \epsilon$.
Furthermore, let $\{\theta^*_i\}_{i \in [n]}$ denote decisions taken according to $x$, a feasible solution to $\lp(\B\cdot (1-\epsilon))$. Formally, $\{\theta^*_i\}_{i \in [n]}$ are random decisions such that when the $i$-th request is of type $k$, we set $\theta^*_i = \theta$ with probability $x_{i, k, \theta} / p_{i, k}$.

We decompose algorithm's total value into  utility and revenue: 
\begin{align*} 
 \alg ~=~ \sum_{i = 1}^\stoptime \val_i(\theta_i) \, &= \,\, \underbrace{\sum_{i = 1}^\stoptime \Big( \val_i(\theta_i) - \langle\bprice_i, a_i(\theta_i)\rangle \Big)}_{\mathsf{Utility}} \,\, + \,\, \underbrace{\sum_{i = 1}^\stoptime \langle\bprice_i, a_i(\theta_i)\rangle}_{\mathsf{Revenue}} \\
 &\geq~ \sum_{i = 1}^\stoptime \Big( \val_i(\theta^*_i) - \langle\bprice_i, a_i(\theta^*_i)\rangle \Big) +  \sum_{i = 1}^\stoptime \langle\bprice_i, a_i(\theta_i)\rangle, \notag 
\end{align*}
where the inequality follows from the fact that the decisions $\{\theta_i\}_{i \in [n]}$ are best-responses (greedy) to the incoming requests.
On taking expectations and rewriting,  
\begin{align}
    \E[\alg] \,\, = \,\, \E\Big[\sum_{i = 1}^\stoptime \val_i(\theta_i)\Big] \,\, \geq \,\, \E\Big[\sum_{i = 1}^\stoptime  \val_i(\theta^*_i)\Big] \,\, - \,\, \E \Big[\sum_{i = 1}^\stoptime   \langle \bprice_i, a_{i}(\theta^*_i) - a_i(\theta_i)  \rangle\Big]. \label{eq:for-augmentation}
\end{align}

The first term on the right hand side can be  lower bounded as
\begin{align}
    \E\Big[ \sum_{i = 1}^\stoptime \val_i(\theta^*_i)\Big] ~=~ \sum_{i = 1}^n \E\Big[\val_i(\theta^*_i) \cdot \one[\stoptime \geq i] \Big]
    ~&=~\sum_{i = 1}^n \E\left[\val_i(\theta^*_i)\right] \cdot \pr[\stoptime \geq i] \notag \\
    & \geq~ \sum_{i = 1}^n \Big(1 - \pr(\EA)\Big) \cdot \E\left[\val_i(\theta^*_i)\right]
    ~\geq~ (1 - \epsilon)\cdot \objx, \notag
 \end{align}
where in the second equality we used independence since $\one[\stoptime \geq i]$ only depends on the requests from $1$ to $i-1$ and $\val_i(\theta^*_i)$ only depends on the $i$-th request. 
The first inequality uses $\pr(\stoptime \geq i) \geq \pr(\stoptime = n) \geq 1 - \epsilon$, and the second (and final) inequality uses the fact that the decisions $\{\theta^*_i\}$ correspond to the solution $x$, whose value we denote by $\objx$. So far, we have shown:

\begin{equation}\label{eq:alg-lb-1}
    \E[\alg] \geq (1-\epsilon)\cdot \objx \,\, - \,\, \underbrace{\E \Big[\sum_{i=1}^\stoptime \langle \bprice_i, a_i(\theta^*_i) - a_{i}( \theta_i) \rangle\Big]}_{\loss}, 
\end{equation}
where the last term in the above expression corresponds to the ``loss in revenue'' that arises from making the best-response decisions (compared to fractional solution $x$).

\medskip
\noindent {\bf Bounding the Loss in Revenue.} We start by rewriting the loss. 
Since the price vector $\bprice_i$ is independent from the request $\gamma_i$ and the random decision $\theta^*_i$, we have
\begin{align}
     \loss \,\, :=\,\, \E \Big[\sum_{i=1}^\stoptime \langle \bprice_i, a_i(\theta^*_i) - a_{i}( \theta_i) \rangle\Big] \,\,
     &= \,\,  \E_{\gamma_1, \ldots, \gamma_n} \Big[ \sum_{i=1}^\stoptime\left\langle \bprice_i,  \E_{\theta^*_i}\left[a_{i}(\theta^*_i)\right] - a_i(\theta_i) \right\rangle\Big] \notag \\
    &=  \,\, \E_{\gamma_1, \cdots, \gamma_n}\Big[ \sum_{i = 1}^\stoptime   \sum_{j = 1}^m \price_{i, j} \cdot \big( \csoln_{i,j}\supx - \calg_{i,j}\big)\Big]. \label{eq:independent}
\end{align}

Now, to bound the $\loss$, it suffices to show that the loss is small for  every  $j \in [m]$.

\begin{Claim} \label{claim:RevLoss}
For any fixed sequence of requests $\gamma_1, \ldots, \gamma_n$, the  loss    for  every resource $j \in [m]$:
    \begin{equation}\label{eq:key}
    \sum_{i = 1}^\stoptime \price_{i, j} \cdot (\csoln_{i, j}\supx - \calg_{i, j}) ~\leq~ \frac{3\initprice}{\delta}  + 5\epsilon \cdot \sum_{i = 1}^\stoptime \price_{i, j} \cdot \calg_{i, j}.
    \end{equation}
\end{Claim}
  Before proving \Cref{claim:RevLoss}, we use it to complete the proof of \Cref{thm:exp-pricing-estimates}. Summing  \eqref{eq:key}  over all $j \in [m]$ gives
    \begin{align*}
    \loss ~ &\leq~   \frac{3m\initprice}{\delta} + 5\epsilon \cdot \E\Big[\sum_{i = 1}^\stoptime \sum_{j = 1}^m \price_{i, j} \cdot \calg_{i, j}\Big]  \\
    &\leq~   \frac{3m\initprice}{\delta} + 5\epsilon \cdot \E\left[\alg\right] ~~\leq~~  \frac{3m\initprice}{\delta} + 5\epsilon \cdot \opt,
\end{align*}
where the second inequality  uses $ \alg =  \sum_{i = 1}^\stoptime   \val_i(\theta_i) \geq \sum_{i = 1}^\stoptime \sum_{j = 1}^m \price_{i, j} \cdot \calg_{i, j}  $ since the utilities are always non-negative due to the null action. Combining this inequality with \eqref{eq:alg-lb-1} gives
\begin{align*}
    \E[\alg] ~&\geq~ (1-\epsilon)\cdot \objx - 5\epsilon \cdot \opt - \frac{3m\initprice}{\delta} \\
    &\geq~ (1-\epsilon)\cdot\objx - 5\epsilon \cdot \opt - \frac{3\epsilon B_j}{2n} \cdot \esopt ~~\geq~~ \objx - 7\epsilon \cdot \opt,
\end{align*}
where the last inequality uses $B_j \leq n$.
\end{proof}

Thus, it only remains to prove \Cref{claim:RevLoss}. The following lemma is our main tool in its proof.

\begin{Lemma}[No-Regret to Zero Price] \label{lem:noRegret}
Consider an online setting where in the $i$-th step we play a price $\accprice_{i, j} \geq 0$ for  item $j$. On playing price $\accprice_{i, j}$, we get revenue  $\accprice_{i, j} \cdot r_{i,j}$ for some adversarially chosen $r_{i,j} \in [-1,1]$ (denote $r_{i,j}^+ := \max\{0, r_{i, j}\}$). Then, playing exponential prices   $\accprice_{i, j} = \initprice \cdot \exp\left(\delta \cdot \sum_{\ell \, < \, i}  r_{\ell,j} \right)$ for some $\delta \in (0,1/2)$ and $\initprice >0$ gives  us total revenue 
\[ 
         \sum_{i = 1}^\stoptime \accprice_{i, j} \cdot r_{i,j} ~\geq ~ - \frac{2\initprice}{\delta} - 4\delta \cdot \sum_{i = 1}^\stoptime \accprice_{i, j} \cdot r_{i,j}^+ .
        \] 
\end{Lemma}
 At first glance, the R.H.S. appears to be a large negative quantity for a small value of $\delta$, but we will set $\initprice \ll \delta \ll 1$ where the R.H.S. approaches $0$.
If the prices $\accprice_{i, j}$ were being played using an online learning algorithm such as  Hedge or Multiplicative-Weights then something like \Cref{lem:noRegret} would immediately follow by applying the no-regret property w.r.t.\,the fixed $\mathbf{0}$ price vector.
 However, as mentioned in \Cref{sec:results-single-sample}, playing Exponential Prices is crucial in our analysis since it allows us to argue that, with high probability, the Exponential
Pricing algorithm never exhausts the budget of any resource (formalized in \Cref{lemma:prob-ea-small}).
In the next subsection we will prove that the exponential pricing algorithm also satisfies such a no-regret w.r.t.\,the $\mathbf{0}$ price vector. 
But first, we use this property to complete the proof of \Cref{claim:RevLoss}.

\begin{proof} [Proof of \Cref{claim:RevLoss}.]
We start by applying \Cref{lem:noRegret} with 
$r_{i,j} = \big( \calg_{i,j} - \csoln_{i, j}\supx  \big)$ to get 
\begin{align} \label{eq:noRegretApplied}
 \sum_{i = 1}^\stoptime \accprice_{i, j} \cdot \big(\csoln_{i, j}\supx - \calg_{i, j} \big) ~&\leq~  \frac{2\initprice}{\delta}  + 4\delta \cdot \sum_{i = 1}^\stoptime \accprice_{i, j} \cdot  \big( \calg_{i,j} - \csoln_{i, j}\supx  \big)^+ \notag \\
~&\leq~ 
 \frac{2\initprice}{\delta}  + 4\delta \cdot \sum_{i = 1}^\stoptime \accprice_{i, j}  \cdot \calg_{i, j}.
\end{align}
This nearly proves \Cref{claim:RevLoss} but a crucial difference is that in this inequality the prices $\accprice_{i, j}$ are adjusted according to $r_{i,j}=\big( \calg_{i,j} - \csoln_{i, j}\supx  \big)$ whereas for \eqref{eq:key} we need to adjust the prices $\price_{i, j}$  using the estimates  with $r_{i,j}= \big( \calg_{i,j} - \widehat{\cons}_{i, j}  \big)$.

Fortunately, this can be fixed using the fact that our estimates are good (\Cref{def:goodestimates}) and satisfy
        $ \big|\sum_{\ell \leq i}\widehat \cons_{\ell, j} - \sum_{\ell\leq i}\cons_{\ell, j}^*\supx\big| \,\, \leq \,\, \frac{\epsilon^2 \cdot B_j}{16 \log (nm\beta/\epsilon)}$ 
for all $i, j$. Combining this  with the fact that
\[ \textstyle
 {\price_{i, j}} = {\accprice_{i, j}} \cdot \exp\Big(\delta \cdot \sum_{\ell<i} \big(\csoln_{\ell, j}\supx - \widehat{\cons}_{\ell, j} \big) \Big)
\]
gives 
$1-\epsilon \leq  \frac{\price_{i, j}}{\accprice_{i, j}} \leq 1+\epsilon$ since $1+x \leq e^x \leq 1 + 2x$ for any $x \in [0, 1]$. This implies
 \begin{align*}
         \sum_{i = 1}^\stoptime \accprice_{i, j} \cdot \big( \csoln_{i, j}\supx -  \calg_{i,j} \big) \,\,  \geq  \,\, \frac{1}{1+\epsilon}\big(\sum_{i = 1}^\stoptime \price_{i, j} \cdot \csoln_{i, j}\supx\big) - \frac{1}{1-\epsilon}\big(\sum_{i = 1}^\stoptime \price_{i, j} \cdot \calg_{i,j}\big). 
    \end{align*}
Combining this with \eqref{eq:noRegretApplied} gives
\begin{align*}
      \frac{1}{1+\epsilon}\big(\sum_{i = 1}^\stoptime \price_{i, j} \cdot \csoln_{i, j}\supx\big)  - \frac{1}{1-\epsilon}\big(\sum_{i = 1}^\stoptime \price_{i, j} \cdot \calg_{i,j}\big) ~&\leq~   \frac{2\initprice}{\delta}  + 4\delta  \sum_{i = 1}^\stoptime \accprice_{i, j}  \cdot \calg_{i, j} \\
     ~&\leq~  \frac{2\initprice}{\delta}  + \frac{\epsilon}{1-\epsilon} \sum_{i = 1}^\stoptime \price_{i, j}  \cdot \calg_{i, j},
\end{align*}
where the last inequality uses  $4\delta \leq \epsilon$ when $B_j \geq {32 \log (nm\beta/\epsilon)}/{\epsilon^2}$. Multiplying both sides of the above inequality by $(1 + \epsilon)$, rearranging, and using  $\epsilon \leq 1/2$, 
    \[
     \sum_{i = 1}^\stoptime \price_{i, j} \cdot (\csoln_{i, j}\supx - \calg_{i, j}) ~\leq~ \frac{3\initprice}{\delta}  + 5\epsilon \cdot \sum_{i = 1}^\stoptime \price_{i, j} \cdot \calg_{i, j},    
    \]
which completes the proof of  \Cref{claim:RevLoss}.
\end{proof}

Next, we prove the missing  \Cref{lemma:prob-ea-small} and  \Cref{lem:noRegret}.


\subsection{No-Regret Property and the Probability of Excess Allocation} \label{sec:NoRegretAndMissingProofs}

In this section, we finish the missing proofs of \Cref{lemma:prob-ea-small} and \Cref{lem:noRegret}. First, we prove the no-regret property in \Cref{lem:noRegret}. Recall, in this online problem we play a price $\accprice_{i, j} \geq 0$ in the $i$-th step and get   revenue $\accprice_{i, j} \cdot r_{i,j} $ from the   $j$-th item where $r_{i,j} \in [-1,1]$; note that $r_{i,j}$  could be negative. We want to lower bound the total revenue of the exponential pricing rule that sets prices as follows.
\[
\textstyle   \accprice_{i, j} \gets \initprice \cdot \exp\left(\delta \cdot \sum_{\ell \, < \, i}  r_{\ell,j} \right).
\]

\begin{proof}[Proof of \Cref{lem:noRegret}.]
We start by observing that
\begin{align*}
    \accprice_{i+1, j} - \accprice_{i, j}~=~ \accprice_{i, j} \cdot\left(\exp \left(\delta \cdot r_{i,j} \right) - 1\right) 
    &~\leq~ \accprice_{i,j} \cdot \left(\delta \cdot r_{i,j} + (\delta)^2 \cdot  r_{i,j}^2\right) \\
    &~\leq~ \accprice_{i,j} \cdot \left(\delta \cdot  r_{i,j} + (\delta)^2   \cdot r_{i,j} \cdot \mathsf{sign}( r_{i,j})\right),
\end{align*}
where the first inequality follows from $e^x - 1 \leq x + x^2$ for $x = \delta \cdot r_{i,j}  \in [-1/2, 1/2]$  and the second inequality uses $|r_{i,j}| \leq 1$. 
On re-arranging this inequality,  we get
\begin{align*}
\accprice_{i, j} \cdot r_{i,j}~&\geq~ \frac{1}{\delta}  \cdot (\accprice_{i+1,j} - \accprice_{i,j}) -  \delta \cdot \accprice_{i,j} \cdot r_{i,j}\cdot \mathsf{sign}(r_{i,j}).
\end{align*}
Subtracting $\delta \cdot \accprice_{i, j} \cdot r_{i,j}$ from both sides of the above inequality, 
\begin{align*}
    (1 - \delta) \cdot \accprice_{i, j} \cdot r_{i,j} ~&\geq~ \frac{1}{\delta}  \cdot  (\accprice_{i+1,j} - \accprice_{i,j}) -  2\delta \cdot \accprice_{i,j} \cdot r_{i,j}^+,
\end{align*}
where we use $1 + \mathsf{sign}(r_{i,j}) = 0$ when $r_{i,j} < 0$.
Summing the above inequality over all $i \in [\stoptime]$ gives
\begin{align*}
(1 - \delta)\cdot \sum_{i \in [\stoptime]} \accprice_{i, j} \cdot r_{i,j} ~&\geq~ \frac{1}{\delta}  \cdot  \sum_{i \in [\stoptime]}(\accprice_{i+1,j} - \accprice_{i,j}) - 2\delta \sum_{i \in [\stoptime] }\accprice_{i,j} \cdot r_{i,j}^+\\
    ~&=~ \frac{1}{\delta} \cdot (\accprice_{\stoptime+1,j} - \accprice_{1,j}) - 2\delta \sum_{i \in [\stoptime] }\accprice_{i,j} \cdot r_{i,j}^+ 
    \quad \geq \quad  -\frac{\initprice}{\delta}  - 2\delta \sum_{i \in [\stoptime] }\accprice_{i,j} \cdot r_{i,j}^+.
\end{align*}
Finally, using $\frac{1}{1 - \delta} \leq 2$ completes the proof of the lemma.
\end{proof}

Finally, we bound the probability of excess allocation and prove \Cref{lemma:prob-ea-small}.

\begin{proof}[Proof of \Cref{lemma:prob-ea-small}.]

We first recall the setting of \Cref{lemma:prob-ea-small}. 
Recall that event $\EA$ corresponds to the situation that
Algorithm~\ref{alg:prsp} terminates for some $i \in [n]$ and $j \in [m]$, 
and our goal is to show that $\pr(\EA) \leq \epsilon$.
Furthermore, since  
$\esopt$ and   $\big\{\widehat{\cons}_{i, j}\big\}$   are \emph{good estimates} of $\opt$ and  $\big\{\cons^*_{i, j}{\supx}\big\}$, we have (by \Cref{def:goodestimates}) that $\pr_{(\gamma_1, \cdots, \gamma_n)}\big[\sum_{i \in [n]} \one \big[\max_{\theta \in \Theta_i} \val_i(\theta) \geq \beta \cdot \esopt\big] > \frac{10}{\epsilon}\big] \leq \epsilon$; i.e.,   the number of high-value requests with respect to $\esopt$ is $O(1/\epsilon)$.
We will argue that under event $\EA$,
we have $\sum_{i \in [n]} \one \left[\max_{\theta \in \Theta_i} \val_i(\theta) \geq \beta \cdot \esopt\right] > \frac{10}{\epsilon}$, 
which, given our assumption occurs with probability  at most $\epsilon$, thus bounding $\pr(\EA)$ by $\epsilon$.
Towards this end, we say that the $i$-th request
is a \emph{low value} (resp., \emph{high value}) request
if $\max_{\theta \in \Theta_i} \val_i(\theta)$ is less than (resp., at least) $\beta \cdot \esopt$.

Now suppose that event $\EA$ is triggered after request $\stoptime$ for resource $j^*$, and that
$|\{i \in [n]:i \text{ is high value} \}| \leq \frac{10}{\epsilon}$.
We will show that, in this case, the algorithm obtains total value at least $n \cdot \beta \cdot \esopt$ from  the low value requests. 
This is a contradiction since each low value request can contribute at most $\max_\theta \val_i(\theta) < \beta \cdot \esopt$ to the total value.
We first note that we have 
$\sum_{i \in [\stoptime] \, : \, i \text{ is low value} }\calg_{i, j^*}  \,\, \geq \,\, \frac{1}{2}\epsilon\cdot B_j - \frac{10}{\epsilon} \,\, \geq \,\, 2,$
where the final inequality uses $B_j \geq {32 \log (nm\beta/\epsilon)}/{\epsilon^2}$ for all $j \in [m]$.
Next, we define $\widehat \ell$ to be an index such that
$\sum_{i \in [\widehat \ell, \stoptime]\, :\, i \text{ is low value} }\calg_{i, j^*} \in [1, 2]$: such an index always exists since at
most one unit of a resource can be allocated for a given request. Let $\alg_{low}$ be the total value gained from those low value requests.
Conditioned on $\EA$, and using $\val_i(\theta_i) - \langle \bprice_i, a_i(\theta_i) \rangle \geq 0$ since $\phi \in \Theta_i$, we have
\begin{equation}\label{eq:largeEA2}
 \alg_{low} \,\, \geq \,\, \sum_{i \in [\stoptime] \, : \, i \text{ is low value} } \sum_{j=1}^m \price_{i, j} \cdot \calg_{i, j} ~\geq~ \sum_{i \in [\widehat \ell, \stoptime] \, : \, i \text{ is low value} } \price_{i, j^*} \cdot \calg_{i, j^*}.
\end{equation}
Also, for $i \in [\widehat \ell, s]$, 
 \begin{align}
 \sum\limits_{\ell < i } \calg_{\ell, j^*} - \sum\limits_{\ell < i}\widehat{\cons}_{ \ell, j^*}  
~\geq~  \sum\limits_{\ell \leq \stoptime } \calg_{\ell, j^*} - \sum\limits_{\ell \leq s}\widehat{\cons}_{ \ell, j^*} - \sum\limits_{i \leq \ell \leq \stoptime} \calg_{i, j^*}
~\geq~  \frac{\epsilon}{2}   B_j - 2 - \frac{10}{\epsilon}  ~ \geq ~\frac{\epsilon}{4}   B_j,\label{eq:largeEA3}
\end{align}
where the second inequality uses the termination condition, the assumption that there are no more than $\frac{10}{\epsilon}$ high value requests and that $\sum_{i \in [\widehat \ell, \stoptime]\, :\, i \text{ is low value} }\calg_{i, j^*} \leq 2$. 
The final inequality above is true when $B_j \geq 44\epsilon^{-2}$. 
Substituting~\eqref{eq:largeEA3} in the definition of $\price_{i,j^*}$ implies  $\lambda_{i, j^*} \geq \initprice \cdot \big( \delta \cdot \frac{1}{4}\epsilon B_j\big) = \initprice \cdot \frac{n^2 m^2 \beta^2}{\epsilon^2}$ for $i \in [\widehat{\ell}, \stoptime]$. 
Using this in \eqref{eq:largeEA2} gives
\[
\alg_{low} ~\geq~ \initprice \cdot \frac{n^2m^2\beta^2}{\epsilon^2} \cdot \sum_{i \in [\stoptime] \, : \, i \text{ is low value} } \cons_{i, j^*} ~\geq~ \initprice \cdot \frac{n^2m^2\beta^2}{\epsilon^2} \cdot 1 ~\geq~ n \cdot \beta \cdot \esopt,
\]
which is a contradiction, since there must be $\alg_{low} < n \cdot \beta \cdot \esopt$.
\end{proof}

\section{Learning Good Estimates using a Single Sample}\label{sec:single-sample}

In this section, we design a single sample online resource allocation algorithm that achieves value $(1 - O(\epsilon)) \cdot \opt$ and prove \Cref{thm:intro-main} by using \Cref{thm:exp-pricing-estimates} from the previous section.

To apply \Cref{thm:exp-pricing-estimates}, we need two kinds of estimates: $\esopt$ and $\{\esconsx\}_{i \in [n], j \in [m]}$, which  are good (see \Cref{def:goodestimates}) with respect to some feasible solution $x$ to  $\lp\left(\B\cdot(1-\epsilon)\right)$ with a large objective value $\obj(x) \geq (1 - O(\epsilon)) \cdot \opt$.
Here, learning a good $\esopt$ will be relatively straightforward: we will use the value of the $(1/\epsilon)$-th largest value in  the sample as $\esopt$ and argue, via concentration, that it satisfies the desired conditions. 
For details, see \Cref{sec:CompletingSingleSample}. 
Our main challenge will be to learn good estimates $\{\esconsx\}_{i \in [n], j \in [m]}$. 
Recall that for these estimates to be good, we need that for all $i \in[n], j \in [m]$,
    \begin{align} \label{eq:goodEst}
    \textstyle 0 \,\, \leq \,\, \widehat{\cons}_{i, j} \,\, \leq \,\,  1 \quad \text{and} \quad 
        \big|\sum_{\ell \leq i}\widehat \cons_{\ell, j} - \sum_{\ell\leq i}\cons_{\ell, j}^*{\supx}\big| \,\, \leq \,\, \frac{\epsilon^2 \cdot B_j}{16 \log (nm\beta/\epsilon)},
    \end{align}
    where $ \cons_{\ell, j}^*{\supx} := \sum_{k = 1}^K \sum_{\theta \in \Theta_{\ell, k}} \left(a_{\ell, k}(\theta)\right)_j \cdot x_{\ell, k, \theta}$.

We note that our goal is to estimate budget consumptions $\{\esconsx\}_{i \in [n], j \in [m]}$  with respect to some solution $x$ with a large $\obj(x)$, rather than to learn $x$ itself. 
Indeed, estimating such an $x$  from a single sample is impossible  since the number of request types $K$ could be arbitrarily large. 
The following is the main result of this section.

{
\begin{Theorem}
\label{lma:step1}
    There exists an algorithm that,
    given a single sample $\{\tilde \gamma_1, \ldots, \tilde \gamma_n\}$ from an online resource allocation instance and assuming all resource budgets $B_j \geq \Omega\big({\log^4(nm/\epsilon)}\cdot{\epsilon^{-6}}\big)$,
    outputs estimates $\{\esconsx\}_{i \in [n], j \in [m]}$, such that, with probability at least $1-2\epsilon$, 
    there exists a solution $x$ of $\lp((1 - \epsilon) \cdot \B)$ with $\obj(x) \geq (1 - 4\epsilon) \cdot \opt$ and the estimates satisfy 
    \eqref{eq:goodEst} with respect to $x$.
\end{Theorem}
}

\noindent \textbf{Proof Overview.}
We start by observing that the lemma is much easier to prove if we were given $\OTild(\epsilon^{-4})$ independent samples from the underlying distributions. This is because by solving the configuration LP relaxation for each sample,  we obtain an unbiased estimator for $\consx$.  As we only need to learn prefix budget consumptions up to  $\OTild(\epsilon^2 \cdot B_j)$  accuracy,  $\OTild(\epsilon^{-4})$ samples suffice, and we can then union bound over every prefix. However, this argument crucially relies on obtaining independent samples.

Given a single sample, a natural idea is to uniformly partition the single sample into $\OTild(\epsilon^{-4})$ parts, and treat each part as an independent sample. That is,  let $\ptt = \{S_1, \ldots, S_D\}$ be a uniformly at random partition of $[n]$ into $D=\OTild(\epsilon^{-4})$   parts where each part $S_d$ has size $n/D$\footnote{Without loss of generality, we assume $n$ is a multiple of $D$. This assumption is feasible because we can arbitrarily add $0$-value requests into the instance.}. 
 Now we can treat  each part $S_d$  as a ``mini-sample'' for estimating  $\consx$ by solving the configuration LP on $S_d$ with budget $\B/D$. However, two  challenges arise with such an approach.  

The first and primary challenge arises from the correlations among different parts, as they must remain disjoint. While we can interpret random partitioning as ``sampling without replacement,” there is no negative correlation between samples of $\consx$ from different parts, making it difficult to apply standard concentration bounds.
Our key idea is to redefine the learning target. Instead of aiming for the ``global” LP optimal solution that depends on all $n$ requests, we define ``local” LP optimal solutions for each part. This approach effectively bypasses correlations, as separate parts no longer interact, allowing us to leverage concentration results with  $\OTild(\epsilon^{-4})$ independent samples. Additionally, we show that when  $\ptt$ is chosen uniformly at random, with high probability the local solutions can be merged to obtain a near-optimal global solution.

The second challenge comes from the inaccuracy of the budget for each mini-sample. After randomly choosing $\ptt$,  each mini-sample needs $\OTild(\epsilon^4 \cdot \B)$ budget in expectation. 
However, because of the standard deviation, some mini-samples may require  more budget than the expectation, leading to an extra loss when only the expected amount of budget is assigned to this mini-sample. 
{To resolve this   challenge, we relax the assumption on the budget to be $B_j \geq \tilde \Omega(\epsilon^{-6})$, instead of $B_j \geq \tilde \Omega(\epsilon^{-2})$, and argue that under the new budget assumption, with high probability, this leads to at most $O(\epsilon)$ loss in the objective value.}

\subsection{Learning Estimates $\{\esconsx\}$ via Sample Partition: Proof of \Cref{lma:step1}}

    We start by defining, for every fixed partition $\ptt = \{S_1, \cdots, S_D\}$ of $[n]$, a feasible solution $x^\ptt$ for  $\lp((1 - \epsilon) \cdot \B)$. Note that this solution depends only on the partition and not on our single sample. 
    To define this solution, consider the following random process.
    \begin{itemize}
        \item For each $i \in S_d$, draw a random $\gamma_i = (\val_i, a_i, \Theta_i)$ from $\D_i$.
        \item Then, solve  \ref{program:notype} for $\{\gamma_i\}_{i \in S_d}$ and         let  $z^{(d)}$ be the corresponding optimal solution.
        \begin{align}
            \text{maximize } \quad \quad  &  \textstyle  \quad\sum_{i \in S_d}  \sum_{\theta \in \Theta_{i}} \val_{i}(\theta) \cdot z_{i, \theta}, \notag \\
           \text{s.t.} \qquad & \textstyle \quad   \sum_{i \in S_d} \sum_{\theta \in \Theta_{i}}  a_{i}(\theta) \cdot z_{i, \theta}  \leq \frac{1 - \epsilon}{D} \cdot \B \tag{$\text{LP}_{\textsc{Sample}}$} \label{program:notype} \\
            \forall i \in S_d, & \textstyle \quad \sum_{\theta \in \Theta_{i}} z_{i,\theta} \leq 1 \notag \\
            \forall i \in S_d,  \theta \in \Theta_{i, k},& \quad   0 \leq z_{i, k, \theta} \leq 1. \notag
        \end{align}

    \end{itemize}
    
    For each $d \in [D]$, let $y^{(d)}$ be a solution of $\lp((1 - \epsilon)\cdot\B / D, \{\D_i\}_{i \in S_d})$ that represents the above random process; i.e., we set  $ y^{(d)}_{i, k, \theta} = \E\big[z^{(d)}_{i, \theta} \cdot \one[\type(i) = k]\big]$, where the function $\type(\cdot)$ represents the type of request $i$.
    In other words, we set $\type(i) = k$ if, and only if, the realization of $\gamma_i$ is $\gamma_{i,k}$. 
    We further define $x^\ptt$ to be the combination of all $\{y^{(d)}\}$; formally, let $x^\ptt_{i, k, \theta} = y^{(d)}_{i, k, \theta}$ for $i \in S_d$. 
    The following claim guarantees that $x^\ptt_{i, k, \theta}$ is a feasible for $\lp((1 - \epsilon) \cdot \B)$.

\begin{Claim} \label{claim:xpFeasibility}
    The solution $x^\ptt$ defined above is feasible for  $\lp((1 - \epsilon) \cdot \B)$.
\end{Claim}
\begin{proof}
    
    To verify the budget constraint, it suffices to check that each mini-solution $y^{(d)}$ uses at most $(1 - \epsilon) \cdot \B/D$ budget, which follows from the fact that
    \begin{align*}
    \textstyle        \sum_{i \in S_d} \sum_{k \in [K]} \sum_{\theta \in \Theta_{i,k}} a_{i,k}(\theta) \cdot y^{(d)}_{i,k,\theta} ~&=~ \textstyle \E_{\{\gamma_i\}_{i \in S_d}}\left[ \sum_{i \in S_d} \sum_{\theta \in \Theta_i} a_i(\theta) \cdot z^{(d)}_{i, \theta} \cdot \sum_{k \in [K]} \one[\type(i) = k]\right] \\
    ~&=~ \textstyle \E_{\{\gamma_i\}_{i \in S_d}}\left[ \sum_{i \in S_d} \sum_{\theta \in \Theta_i} a_i(\theta) \cdot z^{(d)}_{i, \theta}\right] ~\leq~ (1 - \epsilon) \cdot \B/D.
    \end{align*}
    To check that $x^{\ptt}$ satisfies $\sum_{\theta} x^{\ptt}_{i, k, \theta} \leq p_{i, k}$ for all $i \in [n]$ and $k \in [K]$, it's sufficient to check $\sum_{\theta} y^{(d)}_{i, k, \theta} \leq p_{i, k}$ for every $i \in S_d, k \in [K]$. This follows from the fact that
    \begin{align*}
        \textstyle \sum_{\theta \in \Theta_{i,k}} y^{(d)}_{i, k, \theta} ~=~ \E_{\{\gamma_i\}_{i \in S_d}}\left[\sum_{\theta \in \Theta_i}  z^{(d)}_{i, \theta} \cdot \one[\type(i) = k]\right] ~\leq~ \E_{\{\gamma_i\}_{i \in S_d}}\left[1 \cdot \one[\type(i) = k]\right] ~=~ p_{i,k}.
    \end{align*}
    Therefore, $x^\ptt$ is a solution of $\lp((1 - \epsilon) \cdot \B)$.
\end{proof}

Next, we show that $\obj(x^\ptt)$ is large when $\ptt$ (and consequently $x^\ptt$) is chosen uniformly at random.

\begin{Lemma}
\label{lma:xptt}
    Let $D\geq 1$ be an integer that divides $n$. Suppose that, for every resource $j \in [m]$, the budget $B_j$ satisfies $B_j \geq \Omega(\frac{D\cdot \log (nm/\epsilon)}{\epsilon^{2}})$. 
    Then, with probability at least $1 - \epsilon$, a uniformly random  partition $\ptt$ of $[n]$ into parts of size $n/D$ will yield a solution $x^\ptt$ such that $\obj(x^\ptt) \geq (1 - 4\epsilon) \cdot \opt$.
\end{Lemma}

\begin{proof}
    Let $\tildx$ be an arbitrary but fixed solution of $\lp((1 - 3\epsilon) \cdot \B)$, such that $\obj(\tildx) \geq (1 - 3\epsilon) \cdot \opt$. The existence of $\tildx$ is guaranteed by the observation that it's sufficient to scale down the optimal solution of $\lp(\B)$ by a factor of $1 - 3\epsilon$. Our main idea is to show that $\obj(x^\ptt)$ is comparable to $\obj(\tildx) \geq  (1 - 3\epsilon) \cdot \opt$.
    Towards this end, we call a partition $\ptt$ \emph{average} with respect to $\tildx$ if, for every $S_d \in \ptt$, the following condition holds: 
    \[  \textstyle 
    \sum_{i \in S_d} \sum_{k \in [K]} \sum_{\theta \in \Theta_{i,k}} a_{i,k}(\theta) \cdot \tildx_{i,k,\theta} \leq \frac{(1 - 2\epsilon)}{D} \cdot \B.
    \] 
    In other words, each part $S_d$ does not significantly exceed its expected budget consumption.
    The following claim shows that a uniformly random partition $\ptt$ of $[n]$ into parts of size $n/D$ is average with high probability.

    \begin{restatable}{Claim}{clmnottoomuch}\label{clm:nottoomuch} Let $\tildx$ be a solution of $\lp((1-3\epsilon) \cdot \B)$. If $\ptt$ is chosen uniformly at random and $B_j \geq \Omega(\frac{\log (nm/\epsilon) \cdot D}{\epsilon^2})$ for all $j$,     
        then with probability at least $1 - \epsilon$, we have 
        \[  \textstyle 
            \sum_{i \in S_d} \sum_{k \in [K]} \sum_{\theta \in \Theta_{i,k}} a_{i,k}(\theta) \cdot \tildx_{i,k,\theta} \leq \frac{(1 - 2\epsilon)}{D} \cdot \B
        \] 
         for all parts $S_d \in \ptt$; i.e., $\ptt$ is average with respect to $\tildx$. 
    \end{restatable}

      We defer the proof of \Cref{clm:nottoomuch} to \Cref{app:clm-nottoomuch}, as it follows by applying standard concentration inequalities.
     We proceed by assuming that $\ptt$ satisfies \Cref{clm:nottoomuch}, and consider the following random process, which defines a feasible solution $\tildz^{(d)}$ for \ref{program:notype}.
     \begin{itemize}
         \item For each $i \in S_d$, draw $\gamma_i \sim \D_i$.
         \item Check the following condition: 
         \begin{align}  \textstyle 
             \sum_{i \in S_d} \sum_{\theta \in \Theta_i} a_i(\theta) \cdot \sum_{k \in [K]} \one[\type(i) = k] \cdot  \frac{\tildx_{i,k, \theta}}{p_{i,k}} ~\leq~ (1 - \epsilon) \cdot \B/D \label{eq:condition}
         \end{align}
         If Condition \eqref{eq:condition} is satisfied, set $\tildz^{(d)}_{i,\theta} = \frac{\tildx_{i, k, \theta}}{p_{i,k}}$, where $k = \type(i)$; otherwise, set $\tildz^{(d)}_{i,\theta} = 0$.
     \end{itemize}

     The following claim shows that the objective obtained by $\tildz^{(d)}$ is sufficiently high.

     \begin{restatable}{Claim}{clmfollowtildx}
     \label{clm:followtildx}
         Let $\tildx$ be a solution of $\lp((1-3\epsilon) \cdot \B)$, and let $\ptt$ be an average partition with respect to $\tildx$ (per \Cref{clm:nottoomuch}). Then, we have
         \[  \textstyle 
         \E_{\{\gamma_i\}_{i \in S_d}} \left[\sum_{i \in S_d} \sum_{\theta \in \Theta_i} \val_i(\theta) \cdot \tildz^{(d)}_{i, \theta} \right] ~\geq~ (1 - \epsilon) \cdot \sum_{i \in S_d} \sum_{k \in [K]} \sum_{\theta \in \Theta_{i,k}} \val_{i,k}(\theta) \cdot \tildx_{i,k, \theta}
         \]
         for every $d \in [D]$, provided that $B_j \geq \Omega(\frac{\log (nm/\epsilon) \cdot D}{\epsilon^2})$.
     \end{restatable}
     Note that when \Cref{clm:nottoomuch} is satisfied,
     following solution $\tildx$ implies that, in expectation, the realization of $\{\gamma_i\}_{i \in S_d}$ only consumes $(1 - 2\epsilon) \cdot \B/D$ budget. 
     Then, \Cref{clm:followtildx} essentially suggests that Condition \eqref{eq:condition} is satisfied with probability at least $1 - \epsilon$. 
     In other words, with high probability, the consumption of the realized requests does not exceed the expected consumption by a significant amount.
     This, in turn, allows us to lower bound the expected value in terms of the value obtained by $\tildx$.
     We defer the proof of \Cref{clm:followtildx} to \Cref{app:clm-followtildx}, and proceed to complete the proof of \Cref{lma:xptt}.

     Recall that in the random process of solving \ref{program:notype} that defines $y^{(d)}$,  each random $z^{(d)}$ is solved optimally. 
     Therefore, the objective value given by $z^{(d)}$ is at least the value of $\tildz^{(d)}$, i.e.,
     \begin{align*} \textstyle 
         \obj(y^{(d)}) ~=~ \E_{\{\gamma_i\}_{i \in S_d}} \left[\sum_{i \in S_d} \sum_{\theta \in \Theta_i} \val_i(\theta) \cdot z^{(d)}_{i, \theta} \right] ~\geq~ \E_{\{\gamma_i\}_{i \in S_d}} \left[\sum_{i \in S_d} \sum_{\theta \in \Theta_i} \val_i(\theta) \cdot \tildz^{(d)}_{i, \theta} \right].
     \end{align*}
     When \Cref{clm:followtildx} holds, we get
     \[  \textstyle 
     \obj(y^{(d)}) ~\geq~ (1 - \epsilon) \cdot \sum_{i \in S_d} \sum_{k \in [K]} \sum_{\theta \in \Theta_{i,k}} \val_{i,k}(\theta) \cdot \tildx_{i,k, \theta}.
     \]
     Summing the above inequality for all $d \in [D]$ gives
     \[ \textstyle 
     \obj(x^\ptt) ~=~ \sum_{d \in [D]} \obj(y^{(d)}) ~\geq~ (1 - \epsilon) \cdot \obj(\tildx).
     \]
     Therefore, the objective of $x^\ptt$ is at least $(1 - \epsilon) \cdot (1 - 3\epsilon) \cdot \opt \geq (1 - 4\epsilon) \cdot \opt$ when \Cref{clm:nottoomuch} holds, i.e., when $\ptt$ is chosen uniformly at random.
     Thus, with probability at least $1 - \epsilon$, we have $\obj(x^\ptt) \geq (1 - 4\epsilon) \cdot \opt$, which finishes the proof of \Cref{lma:xptt}.
\end{proof}
The final step in proving \Cref{lma:step1} is to show that a single sample from each request distribution is sufficient to concentrate around $x^\ptt$. 
Formally, we give the following lemma.

\begin{Lemma}
    \label{lma:xptt-concentrate}
    Let $D \geq 1$ be an integer that divides $n$. Given a single sample $\{\tilde \gamma_1, \ldots, \tilde \gamma_n\}$ from each of the $n$ request distributions, and a partition $\ptt$ of $[n]$ into parts of size $n/D$, there exists an algorithm that outputs $\{\widehat \cons^{\ptt}_{i, j}\}_{i \in [n], j \in [m]} \in [0, 1]$ such that with probability at least $1 - \epsilon$,   we have for every $i, j$  that
    \[ \textstyle 
    \left|\sum_{\ell \leq i} \widehat \cons^{\ptt}_{\ell, j} - \sum_{\ell \leq i}  \cons^*_{\ell, j}{(x^\ptt)}\right| \leq \epsilon_D \cdot B_j
    \]
    for every $i \in [n]$ and $j \in [m]$, as long as $\epsilon_D \geq \sqrt{\frac{4\log (nm/\epsilon)}{D}}$.
\end{Lemma}

\begin{proof}
We start by considering the following algorithm.
\begin{itemize}
    \item For each $S_d \in \ptt$, solve \ref{program:notype} on requests $\{\tilde \gamma_i\}_{i \in S_d}$, and let $\bar z^{(d)}$ be the corresponding optimal solution. As discussed in \Cref{sec:LPandGoodEstimates}, we can solve this LP in $\poly(n,m)$ time (since $K=1$ here) given demand-oracle access to the $n$ requests.

    \item For each $i \in [n]$ and $j \in [m]$, compute
    \[ \textstyle 
    \widehat \cons^{\ptt}_{i, j} ~:=~ \sum_{d \in [D]} \one[i \in S_d] \cdot \sum_{\theta \in \tilde \Theta_i} \left(\tilde a_{i}(\theta)\right)_j \cdot \bar z^{(d)}_{i, j} \enspace .
    \]
\end{itemize}
We will show that for each $i \in [n], j \in [m]$, we have
\begin{equation}\label{eq:xptt-concentrate-1}
\textstyle 
\pr\left(\left|\sum_{\ell \leq i} \widehat \cons^{\ptt}_{\ell, j} - \sum_{\ell \leq i} \cons^*_{\ell, j}{(x^\ptt)}\right| > \epsilon_D \cdot B_j \right) ~\leq~ \frac{\epsilon}{n  m}  \enspace .
\end{equation} 
Then, applying the union bound over all $i \in [n]$ and $j\in [m]$ finishes the proof.

We now fix $i\in [n]$ and $j \in [m]$, and prove \Cref{eq:xptt-concentrate-1}. First, we define random variables $A^{(d)}_{\leq i, j} := \sum_{\ell \in S_d: \, \ell \leq i} \sum_{\theta \in \tilde \Theta_\ell} \left(\tilde a_{\ell}(\theta)\right)_j \cdot \bar z^{(d)}_{\ell, j}$, where the randomness stems from the realization of requests $\{\tilde \gamma_\ell\}_{\ell \in S_d}$. 
Since $\bar z^{(d)}$ is feasible for \ref{program:notype} for requests  $\{\tilde \gamma_\ell\}_{\ell \in S_d}$, we have
\[ \textstyle 
A^{(d)}_{\leq i, j} ~\leq~ \sum_{\ell \in S_d} \sum_{\theta \in \tilde \Theta_\ell} \left(\tilde a_{\ell}(\theta)\right)_j \cdot \bar z^{(d)}_{\ell, j} ~\leq~ \frac{(1 - \epsilon)}{D} \cdot B_j  \enspace . 
\]
On the other hand, we have
\begin{align*} \textstyle 
    \E\left[A^{(d)}_{\leq i, j}\right] ~&=~ \textstyle   \E_{\{\tilde \gamma_\ell\}_{\ell \in S_d}}\left[\sum_{\ell \in S_d: \, \ell \leq i} \sum_{k \in [K]} \one[\type(\ell) = k] \cdot \sum_{\theta \in \Theta_{i,k}}  \left(\tilde a_{\ell}(\theta)\right)_j \cdot \bar z^{(d)}_{\ell, \theta}\right] \\
    ~&=~  \textstyle \sum_{\ell \in S_d: \, \ell \leq i} \sum_{k \in [K]}  \sum_{\theta \in \Theta_{i,k}}  y^{(d)}_{\ell, k, \theta} \cdot \left(\tilde a_{\ell}(\theta)\right)_j \\
    ~&=~ \textstyle  \sum_{\ell \in S_d: \, \ell \leq i}\sum_{k \in [K]}  \sum_{\theta \in \Theta_{i,k}} x^{\ptt}_{\ell, k, \theta} \cdot \left(\tilde a_{\ell}(\theta)\right)_j,
\end{align*}
where the second equality follows from the observation that $\bar z^{(d)}$ and $z^{(d)}$ (the random solution used to define $y^{(d)}$) are uniformly distributed.
Taking a sum over all $d \in [D]$, we get
\[ \textstyle 
\sum_{d \in [D]} \E\left[A^{(d)}_{\leq i, j}\right] ~=~ \sum_{\ell \leq i} \sum_{k \in [K]} \sum_{\theta \in \Theta_{i,k}} x^\ptt_{\ell, k, \theta} \cdot \left(\tilde a_{\ell}(\theta)\right)_j ~=~ \sum_{\ell \leq i} \cons^*_{\ell, j}{(x^\ptt)},
\]
i.e., $\sum_{\ell \leq i} \widehat \cons^{\ptt}_{\ell, j} := \sum_{d \in [D]} A^{(d)}_{\leq i, j}$ is an unbiased estimator of $\sum_{\ell \leq i} \cons^*_{\ell, j}{(x^\ptt)}$. By Hoeffding's Inequality (\Cref{Hoeffding}), we have
\begin{align*} \textstyle 
    \pr\left(\left|\sum_{\ell \leq i} \widehat \cons^{\ptt}_{\ell, j} - \sum_{\ell \leq i} \cons^*_{\ell, j}{(x^\ptt)}\right| > \epsilon_D\cdot B_j\right) ~\leq~ 2\exp\left(-\frac{2\epsilon^2_D B^2_j}{D \cdot (1 -\epsilon)^2 \cdot B^2_j/D^2}\right) ~\leq~ \frac{\epsilon}{nm},
\end{align*}
where the last inequality holds when $\epsilon_D \geq \sqrt{\frac{4\log (nm/\epsilon)}{D}}$.
\end{proof}

We now complete the proof of \Cref{lma:step1} using \Cref{claim:xpFeasibility}, \Cref{lma:xptt}, and \Cref{lma:xptt-concentrate}.

\begin{proof}[Proof of \Cref{lma:step1}]
We first generate a random partition $\ptt$ according to \Cref{lma:xptt}   with $D = {1024 \log^3(nm/\epsilon)}/{\epsilon^4}$. 
This choice of $D$ and $B_j \geq \Omega({\log^4(nm/\epsilon)}/{\epsilon^6})$ ensures that  the condition $B_j \geq \Omega({D\cdot \log (nm/\epsilon)}/{\epsilon^2})$ in \Cref{lma:xptt} is satisfied, so the lemma guarantees that with probability at least  $1-\epsilon$ we have $\obj(x^\ptt) \geq (1 - 4\epsilon) \cdot \opt$. 
Moreover, \Cref{claim:xpFeasibility} implies that $x^\ptt$  is a feasible solution for $\lp((1 - \epsilon) \cdot \B)$. 
Thus, to prove \Cref{lma:step1}, it suffices  learn good estimates  $\{\widehat \cons_{i, j}\}_{i\in [n], j \in [m]}$  (as in \eqref{eq:goodEst}) w.r.t.\,$x^\ptt$. 
Such good estimates can be obtained for partition $\ptt$ with probability at least $1-\epsilon$ using the single sample algorithm in \Cref{lma:xptt-concentrate} with  $\epsilon_D = \sqrt{{4\log (nm/\epsilon)}/{D}} = {\epsilon^2}/{(16 \log (nm/\epsilon))}$. 
Finally, applying the union bound over the guarantees of \Cref{lma:xptt}  and \Cref{lma:xptt-concentrate}, we can conclude that with probability at least $1 - 2\epsilon$, all the conditions in \Cref{lma:step1} hold. 
\end{proof}

\subsection{Learning $\widehat \opt$ and Completing the Proof of \Cref{thm:intro-main}}
\label{sec:CompletingSingleSample}

Now we complete the proof of our main \Cref{thm:intro-main} by combining the  estimated prefix budget consumptions from the last subsection with  Exponential Pricing  from the last section.

\begin{Theorem}[Formal version of \Cref{thm:intro-main}]
     Given a single sample from each of the $n$ request distributions for online resource allocation with \emph{non-identical} distributions, there exists a $\big(1-O(\epsilon)\big)$-approximation posted pricing algorithm, provided that  every resource has $\Omega\big({\log^4(nm/\epsilon)}/{\epsilon^{6}}\big)$ budget.  
\end{Theorem}

\begin{proof}
    We assume $\epsilon \leq 0.1$, otherwise the theorem trivially holds. We  will also need the following claim regarding learning an estimate of the optimal value.
    
\begin{Claim}
    \label{lma:step2} Given a single sample $\{\tilde \gamma_1, \ldots, \tilde \gamma_n\}$ of the instance and provided that $\epsilon \leq 0.1$, there exists an algorithm that outputs $\esopt$, which satisfies with probability at least $1 - 2\epsilon$ the following:
    \begin{itemize}
        \item $\esopt \leq \opt$.
        \item $\pr_{(\gamma_1, \cdots, \gamma_n)}\left(\sum_{i \in [n]} \one\left[\max_{\theta \in \Theta_i} \val_i(\theta) > \esopt\right] > 10\epsilon^{-1} \right) ~\leq~ \epsilon.$
    \end{itemize}
\end{Claim}

    Before proving the claim, we use it to complete the proof of the theorem.  
    Combining \Cref{lma:step1}  from the last subsection and \Cref{lma:step2} using union bound, we get    with probability at least $1 - 4\epsilon$   good estimates $\esopt$ and $\{\widehat \cons_{i,j}\}_{i \in [n], j\in [m]}$ that satisfy \Cref{def:goodestimates} with $\beta = 1$ w.r.t. solution $x$ of $\lp((1 - \epsilon) \cdot \B)$ which has $\obj(x) \geq (1 - 4\epsilon) \cdot \opt$.  Now applying \Cref{thm:exp-pricing-estimates} from the last section with these estimates gives a posted pricing algorithm with expected value at least 
    \[
    (1 - 4\epsilon) \cdot \big((1-\epsilon)(1 - 4\epsilon)\cdot \opt - 7\epsilon \cdot \opt\big)~\geq~ (1 - 16\epsilon) \cdot \opt. \qedhere
    \]
\end{proof}
 
Next, we prove the missing claim.
\begin{proof}[Proof of \Cref{lma:step2}]
    We assume that for all $i \neq i'$, we have $\max_{\theta} \val_i(\theta) \neq \max_{\theta} \val_{i'}(\theta)$.
    This assumption is without loss of generality since we can add an arbitrarily small noise to the value functions, or equivalently, we can break ties uniformly at random.

    Given $\{\tilde \gamma_1, \ldots, \tilde \gamma_n\}$, our algorithm outputs $\esopt$ that satisfies $\sum_{i \in [n]} \one[\max_{\theta \in \tilde \Theta_i} \tilde \val_i(\theta) \geq \esopt] = 3/\epsilon$. 
    We prove \Cref{lma:step2} by showing that the two required conditions are both satisfied with probability at least $1 - \epsilon$ for the value of $\esopt$ defined above. At a high-level, we show that with high probability the number of requests with highest value above $\esopt$ is at least $\epsilon^{-1}$ and at most $6\epsilon^{-1}$. Then, the first argument implies $\esopt \leq \opt$, and the second argument directly covers the second condition in \Cref{lma:step2}.

    We first prove that $\esopt \leq \opt$ with  probability at least $1 - \epsilon$. 
    Let $\eta_1$ be a value that satisfies $\E\big[\sum_{i \in [n]} \one[\max_{\theta \in \Theta_i} \val_i(\theta) > \eta_1]\big] = 1/\epsilon$. Note that it suffices to show $\esopt \leq \eta_1$ with high probability. This is because when $\esopt \leq \eta_1$, we have
    \[ \textstyle
    \opt ~\geq~ \eta_1  ~\geq~ \esopt.
    \]
    To see the first inequality, define random variables $X_i = \one[\max_{\theta \in \Theta_i}v_i(\theta) > \eta_1]$, and note that $\sum_{i \in [n]}\E[X_i] = 1/\epsilon$. Then, applying the (one-sided version of) Bernstein's inequality to the mean-zero random variables $\E[X_i] - X_i$ gives
    \[
    \textstyle 
    \pr\left( \sum_{i \in [n]} (\E[X_i] - X_i)  \geq \frac{1}{2\epsilon}\right) ~\leq~ \exp\left( - \frac{1/8\epsilon^2}{\sigma^2 + 1/(3\epsilon)} \right) ~\leq~ 6\epsilon,
    \]
    where the final inequality uses $\sigma^2 = \sum_{i\in[n]}\sigma_i^2 \leq 1/\epsilon$ and $\epsilon \leq 0.1$. Thus, $\opt \geq (1-6\epsilon) \cdot \frac{1}{2\epsilon} \cdot \eta_1 \geq \eta_1$, which again uses $\epsilon \leq 0.1$.
    Further, note that 
    \begin{align*}\textstyle
        \pr\left(\esopt > \eta_1 \right) ~\leq~ \pr\left(\sum_{i \in [n]} (X_i - \E[X_i]) \geq \frac{3}{\epsilon} - \frac{1}{\epsilon}\right) ~\leq~ \exp\left(-\frac{4 /(2\epsilon^2)}{\sigma^2 + 2/(3  \epsilon)} \right) ~\leq~ \epsilon,
    \end{align*}
    which again uses $\sigma^2 = \sum_{i\in[n]}\sigma_i^2 \leq 1/\epsilon$ and $\epsilon \leq 0.1$.
    Therefore, we have $\esopt \leq \opt$ with probability at least $1 - \epsilon$.

    For the second condition, we define $\eta_2$ to be a value that satisfies $\E\big[\sum_{i \in [n]} \one[\max_{\theta \in \Theta_i} \val_i(\theta) > \eta_2]\big] = 6/\epsilon$. 
    Then, we define random variables $Y_i = \one[\max_{\theta \in \Theta_i} \val_i(\theta) > \eta_2]$, and note $\sum_{i \in [n]} \E[Y_i] = 6/\epsilon$. Let $\sigma_2^2 := \E[(Y_i - \E[Y_i])^2]$. 
    Observe that $\sigma_2^2 \leq \sum_{i \in [n]} \E[Y_i] = 6/\epsilon$, and
    that it suffices to show $\esopt \geq \eta_2$ with high probability. 
    This is because when $\esopt \geq \eta_2$, we have
    \begin{align*}\textstyle
        \pr\left(\sum_{i \in [n]} \one\left[\max_{\theta \in \Theta_i} \val_i(\theta) > \esopt\right] > 10\epsilon^{-1} \right) ~&\leq~ \textstyle\pr\left(\sum_{i \in [n]} \one\left[\max_{\theta \in \Theta_i} \val_i(\theta) > \eta_2\right] > 10\epsilon^{-1} \right) \\
        ~&=~\textstyle \pr\left(\sum_{i \in [n]} (Y_i  - \E[Y_i]) > \frac{10}{\epsilon} - \frac{6}{\epsilon}\right) \\
        ~&\leq~ \textstyle \exp \left(-\frac{16/(2\epsilon^2)}{6/\epsilon + 4/(3\epsilon)}\right) ~\leq~ \epsilon,
    \end{align*}
    where the penultimate inequality uses Bernstein's Inequality, and the last inequality holds when $\epsilon \leq 0.1$.
    It remains to show that $\pr(\esopt \geq \eta_2)$ is at least $1 - \epsilon$. This follows by applying  Bernstein's Inequality to variables $Y_{i} - \E[Y_i]$ as follows.
    \begin{align*}\textstyle
        \pr\left(\esopt < \eta_2 \right) ~=~ \pr\left(\sum_{i \in [n]} (Y_i - \E[Y_i]) \leq \frac{3}{\epsilon} - \frac{6}{\epsilon}\right) ~\leq~ \exp\left(-\frac{9 /(2\epsilon^2)}{\sigma_2^2 + 3/(3  \epsilon)} \right) ~\leq~ \epsilon,
    \end{align*}
    where the last inequality holds when $\epsilon \leq 0.1$.

    Finally, combining two statements with union bound proves \Cref{lma:step2}.
\end{proof}


\section{Robustness in the Byzantine Secretary Model } \label{sec:robustness}
In this section, we show that the  
Exponential Pricing algorithm is robust to outliers: it achieves a $(1-\epsilon)$-approximation in the {Byzantine Secretary} model.

\subsection{Byzantine Secretary Model and Proof Overview}
We now formally describe an instance of online resource allocation in the Byzantine Secretary model \cite{BGSZ-ITCS20,AGMS-SODA22}.
There are two sets of requests:  $\numg$ green (stochastic or good) requests and  
$\numr$ red (adversarial or rogue) requests. 
Each green request $i$, represented by $\gamma_i = (\val_i, a_i, \Theta_i)$, arrives at a uniformly random time $t_i \sim U[0, 1]$, while
each red request $i$, represented by  $\gamma^\mathsf{R}_i = (\val^\mathsf{R}_i, a^\mathsf{R}_i, \Theta^\mathsf{R}_i)$, arrives at an adversarially chosen time $t^\mathsf{R}_i \in [0, 1]$. 
We assume that the adversary is oblivious,
so the model dynamics are as follows:
The adversary first decides the set of requests $\gamma _1, \ldots, \gamma _{\numg}$ and $\gamma^\mathsf{R}_1, \ldots, \gamma^\mathsf{R}_{\numr}$, and the arrival times $t^\mathsf{R}_1, \ldots, t^\mathsf{R}_{\numr}$
of the red requests. 
Then each green request is independently assigned a uniformly random arrival time $t _i \sim \text{Unif}[0, 1]$.
When a request arrives, the algorithm does not see the color (green
or red): it just sees the request and the time of arrival.
Our goal is to design a pricing algorithm that makes an irrevocable decision upon receiving a request,  subject to the given budget constraints.

We note that since the red requests are chosen adversarially and have adversarial arrival order, we cannot hope to get the value of the offline optimum on the red requests. 
As mentioned in \Cref{sec:introModel}, the best possible approximation in this case is $\Omega\left(n/B\right)$ 
even for the special case of a single resource.
Consequently, we  compare the performance of 
our algorithm to the expected offline optimum of the green requests. 
Specifically, we define the benchmark for our algorithm to be
\begin{align}
\label{program:ro}
\textstyle
    \opt ~:=~ \max_{\{\theta _i \in \Theta _i\}_{i \in [\numg]}} \sum_{i = 1}^{\numg} \val _i(\theta _i) ~\text{ s.t. } \sum_{i = 1}^{\numg} a _i(\theta _i) \leq \B .
\end{align}
Additionally, 
we assume that an estimate $\widehat{\opt}$ of $\opt$  is provided, with $\esopt \in [\frac{\opt}{\beta}, \opt]$. As argued in \cite{BGSZ-ITCS20}, such an assumption is necessary  since it is impossible to compare to unbounded $\opt$.

The main result of this section is a $(1-\epsilon)$-approximation in the Byzantine Secretary model, greatly improving the $\Omega(1)$-approximation of \cite{AGMS-SODA22} in the same setting.

\begin{Theorem}
\label{thm:romainori}
Given $\epsilon > 0$ and an estimate $\esopt \in \left[\frac{\opt}{\beta}, \opt\right]$, there exists a pricing algorithm for online resource allocation in the Byzantine Secretary model that obtains expected total value at least $\left(1 - O(\epsilon)\right) \cdot \opt$ when the  budget satisfies $B_{j} =  \Omega\left({\log(m\beta/\epsilon)}\cdot {\epsilon^{-2}}\right)$ for every $j \in [m]$. 
\end{Theorem}

\noindent\textbf{Proof Overview.} 
Our idea is to apply  Exponential Pricing to the Byzantine Secretary model. To facilitate this,  in \Cref{sec:disByz} we first transform the original Byzantine Secretary model into a ``discrete" model, converting the continuous time horizon 
$[0,1]$ into a finite number of time slots.

Once the Exponential Pricing algorithm is applied, the proof structure resembles that of  \Cref{thm:exp-pricing-estimates}. A key difference arises in  \eqref{eq:independent}  of that proof, where we rewrote   $\loss$ by moving the expectation inside the dot product. This step is only possible if the price vector  $\bprice_i$ is independent of the realization of request $\gamma_i$. However, 
in the Byzantine Secretary model, since the green requests arrive in random order, the realization of request $\gamma_i$ depends on previous requests, creating correlation with the corresponding price vector. 
To address this, we interpret the randomness in historical data as a ``sampling without replacement'' process and apply concentration inequalities for this type of sampling. Similar to previous work \cite{AGMS-SODA22, GM-MOR16}, we also restart the Exponential Pricing algorithm at the midpoint to further control the effects of dependencies.

\subsection{Discretized Byzantine Secretary Model}
\label{sec:disByz}

To prove \Cref{thm:romainori}, we first  reduce  the problem to solving a ``discretized'' version of the Byzantine Secretary model. 
Then, we show that an algorithm for online resource allocation in the discretized  model can be converted into an algorithm for online resource allocation in the original Byzantine Secretary model with a negligible loss in the value (see Lemma~\ref{lma:oritodisc} for a formal statement).

\medskip
\noindent\textbf{Model.}
The intuition behind discretizing the input is as follows.
Suppose we could 
break the time interval $[0, 1]$ into $T$ smaller intervals where $T \gg \numg + \numr$; 
now the $t$-th interval, or piece, represents the time $\big[(t - 1) \cdot \frac{1}{T}, t \cdot \frac{1}{T}\big)$,
then the intervals
containing green requests should contain 
at most one green request (with probability arbitrarily close to $1$).
Building on this intuition, we formally define an instance of online resource allocation in the discretized Byzantine Secretary model as follows.
Given an $\epsilon > 0$, we have $T > (\numg + \numr) / \epsilon$ time intervals. 
The adversary first decides the set of green requests $\gamma _1, \ldots, \gamma _{\numg}$. 
Then, the adversary selects intervals $R \subseteq [T]$ such that $|R| \leq \numr$, and for each $t \in R$ sets a red request 
$\gamma_t = (\val_t, a_t, \Theta_t)$.
Finally, the set of green requests are permuted uniformly at random amongst the time intervals 
in the set $[T] \setminus R$.
This gives a mapping $\pi : [\numg] \to [T] \setminus R$ which maps each one of the green requests to one of the remaining time intervals.
Without loss of generality, we 
use $\gamma _{i}$ to denote the request 
$\gamma_{\pi(i)} = (\val_{\pi(i)}, a_{\pi(i)}, \Theta_{\pi(i)})$
for $i \in [\numg]$.
For the intervals that are not assigned any request; that is, for $t \in [T] \setminus \left(\{\pi(i): i \in [\numg]\} \cup R\right)$, we set $\gamma_t$ to be 
a dummy request that only contains the null decision $\phi$.

\begin{Lemma}
\label{lma:oritodisc}
Suppose that 
we are given an algorithm for online resource allocation problem in the discretized Byzantine Secretary model
that achieves expected total value at least $\eta \cdot \opt$ when $T > \numg$ (the number of green requests).
Then,  
we can obtain an algorithm for 
online resource allocation problem in the Byzantine Secretary model
that achieves expected total value at least $(1 - \disceps) \cdot \eta\cdot\opt$, where $\opt$ equals the expected offline optimum of the green requests.
\end{Lemma}

\begin{proof}
Let $\A_d$ denote an algorithm for online resource allocation in the discretized Byzantine Secretary model. 
Now, consider an instance of  online resource allocation in the Byzantine Secretary model with $\gamma _1, \cdots, \gamma _{\numg}$ denoting the green requests, and $\opt$ denoting the offline optimum on these green requests.
The main idea is to appropriately select $T$ such that the resulting instance behaves like an instance of the discretized Byzantine Secretary model, and then to run $\A_d$ on this resulting instance. 

\medskip
\noindent{\bf Reduction.} We set $T = n^2/\disceps$ where $n = \numg + \numr$, and consider 
intervals of the form $\big[(t - 1) \cdot \frac{1}{T}, t \cdot \frac{1}{T}\big)$. 
Recall that the adversary first decides the set of requests $\gamma _1, \ldots, \gamma _{\numg}$ 
and 
$\gamma^\mathsf{R}_1, \ldots, \gamma^\mathsf{R}_{\numr}$, 
and the arrivals of the red requests; that is, $t^\mathsf{R}_1, \ldots, t^\mathsf{R}_{\numr}$, where  each $t_i^\mathsf{R} \in [0, 1]$.
Then each green request is independently assigned a uniformly random arrival time $t _i \sim \text{Unif}[0, 1]$.
We say that there is a \emph{conflict} if
a green request 
falls in an interval already containing another request (either red or green).
Note that the probability that a green request results in a conflict is at most $n/T$ (since at most $n$ intervals could contain other requests),
and by the union bound the probability of a conflict is at most $n^2/T$, which equals $\disceps$ by our choice of $T$. 
So, with probability at least $1 - \disceps$, each green request appears by itself in a time interval. 
We note that this is 
equivalent to selecting a random permutation $\pi$ as defined
in the discretized Byzantine Secretary model.
Furthermore, for a time interval that contains more than one red request, we keep only the first red request and discard the remaining. 
Lastly, for any empty interval, we assume that there is a dummy request with only the null decision $\phi$.

\medskip
\noindent{\bf Value.} From the previous discussion, with probability at least $1 - \disceps$, we have an instance of online resource allocation in the discretized Byzantine Secretary model with expected offline optimum of the green requests equal to $\opt$.
Running $\A_d$ gives us \emph{conditional} expected total value at least $\eta \cdot \opt$.
On de-conditioning, we get that the
expected total value is at least $(1 - \disceps)\cdot \eta \cdot \opt$. \end{proof}
Thus, our main result involves showing that there exists a nearly optimal pricing algorithm for online resource allocation in the discretized Byzantine secretary model.

\begin{Theorem}
\label{thm:romain}
Given $\epsilon > 0$, $T > {(\numg + \numr)}/{\epsilon}$, and an estimate $\esopt \in \big[\frac{\opt}{\beta}, \opt\big]$, there exists a pricing algorithm for the online resource allocation in the discretized Byzantine Secretary model that obtains expected total value at least $\left(1 - O(\epsilon)\right) \cdot \opt$ when budget satisfies $B_{j} =  \Omega\left(\frac{\log(m\beta/\epsilon)}{\epsilon^2}\right)$ for every $j \in [m]$, where $\opt$ represents the expected offline optimum of the green requests.
\end{Theorem}

Combining \Cref{thm:romain} with  \Cref{lma:oritodisc}  proves \Cref{thm:romainori}.

\subsection{Exponential Pricing for Discretized Byzantine Secretary}\label{sec:discByzantineSec}

In this section, we give the proof of \Cref{thm:romain}.  
As before, the algorithm uses the expected allocation of a given resource for each request. 
In the case of the Random-Order model (without corruptions), the expected allocation of resource $j \in [m]$ at time $t \in [T]$ of an optimal solution would be $B_j / T$.
Motivated by this, we find setting
$(1-\epsilon) \cdot B_j/T$
for every $t \in [T]$ and every $j \in [m]$ suffices, where the expected consumption is scaled down by a factor of $1 - \epsilon$ to avoid exhausting budget.
The prices are set according to~\eqref{eq:ora-ro-prices}. We note that the scaling factor in \eqref{eq:ora-ro-prices} is $\epsilon$, instead of the $\delta = \OTild(1/(\epsilon \cdot B_j))$  used in \Cref{alg:prsp}. Essentially, this change helps us avoid losing a $\log T$ factor in the large budget assumption. However, the corresponding $\log n$ factor in  \Cref{alg:prsp} is unavoidable due to the learning error from the single sample.

Given the prices, the $t$-th decision
involves playing a best-response decision.
Lastly, we set the following termination condition: if, for any $i \in [T/2]$, there exists $j \in [m]: \price_{t+1, j} = \initprice \cdot \exp\left(\epsilon \cdot \left(\sum_{s\leq t}\calg_{s, j} - t \cdot (1-\epsilon) \cdot B_j/T \right)\right) > \initprice \cdot  (m\beta/\epsilon)^8$, the algorithm stops/terminates. 
The algorithm restarts at $t = T/2 + 1$.

\begin{algorithm}
\caption{\textsc{Exponential-Pricing-Byzantine} ($\mathcal{I}, \epsilon, \widehat{\opt}$)}
\label{alg:ropip}
\begin{algorithmic}[1]
\For{$t = 1, \ldots, \frac{T}{2}$}
\State $(\val_t, a_t, \Theta_t) \gets $ $t$-th request 
\State $\bprice_t \gets $ vector denoting resource prices before $t$-th request arrives;
for $j \in [m]$, we set 
\begin{equation}\label{eq:ora-ro-prices}
\textstyle  \price_{t,j} \gets \initprice \cdot \exp\left(\epsilon \cdot \left(\sum_{s < t}\calg_{s, j} - (t-1)\cdot (1-\epsilon) \cdot B_j/T \right)\right)
\end{equation}
\qquad where  $\initprice \gets \epsilon^5 \cdot \esopt / m^4$
\State Take $\theta_t \in \arg \max_{\theta \in \Theta_t} \left(\val_t(\theta) - \langle \bprice_t, a_t(\theta) -  (1-\epsilon) \cdot \B/T \rangle\right)$, and the algorithm gains $\val_t(\theta_t)$.
\State Set vector $(\calg_{t,1}, \cdots, \calg_{t, m}) = a_t(\theta_t)$.
\If{ $\exists j \in [m]: \,$ $\price_{t+1, j} = \initprice \cdot \exp\left(\epsilon \cdot \left(\sum_{s\leq t}\calg_{s, j} - \frac{t \cdot (1-\epsilon)B_j}{T} \right)\right) > \initprice \cdot  (m\beta/\epsilon)^8$}
\State  \textbf{break}
\EndIf
\EndFor
\State Restart the algorithm for $t = \frac{T}{2} + 1, \ldots, T$.
\end{algorithmic}
\end{algorithm}

Now, we prove \Cref{thm:romain} by showing that \Cref{alg:ropip} is the desired algorithm.

\begin{proof}[Proof of \Cref{thm:romain}] In this proof, we assume $\epsilon \in [0,1/2]$ to be a fixed parameter, since the theorem is true for $\epsilon>1/2$.

To simplify notation, we define the vector 
\[
\ba^* ~=~ (a^*_1, \cdots, a^*_m) ~:=~ \frac{(1 - \epsilon) \cdot \B}{T},
\] which corresponds to the scaled expected budget consumption in each round. 
We first verify that \Cref{alg:ropip} is feasible. To do this, we first show that the first half of the algorithm
uses at most ${B_j}/{2}$ budget for any resource $j \in [m]$. Suppose not, then there exists $t \in [T/2]$ and $j \in [m]$ that satisfies 
    \[
    \sum_{s = 1}^t \calg_{s, j} ~>~ \frac{B_j}{2} ~\geq~ t \cdot a^*_j+ \frac{\epsilon \cdot B_j}{2}.
    \]
Then, we have
\begin{align*}
\price_{t, j} ~&=~ \initprice \cdot \exp \left( \textstyle \epsilon \cdot \left(\Big(\sum_{s = 1}^t \calg_{s, j} - t\cdot a^*_j\Big) - \Big(\calg_{t, j} - a^*_j\Big)\right)\right)\\
&>~ \initprice \cdot \exp \left( \textstyle \epsilon \cdot \left(\frac{\epsilon \cdot B_j}{2} - 1\right)\right) ~\geq~ \initprice \cdot  (m\beta/\epsilon)^8,
\end{align*}
where the last inequality holds when $B_j  \geq  {20\log (m\beta/\epsilon)}/{\epsilon^2}$.
Note that the above inequality leads to a contradiction (see the termination condition), and thus the first half of the algorithm uses budget at most $B_j / 2$ for every $j \in [m]$.

This observation implies that we can 
analyze the second half of \Cref{alg:ropip} in the same way as the first half, i.e., the first and second half of \Cref{alg:ropip} are symmetric. 
Let $\alg_{ro}$ be a random variable denoting the value obtained in the first half of \Cref{alg:ropip} (value collected until ${T}/{2}$). 
Then, it's sufficient to show $\E[\alg_{ro}] \geq (1 - O(\epsilon)) \cdot \frac{\opt}{2}$: the expected total value 
obtained by \Cref{alg:ropip} is $2 \cdot \E[\alg_{ro}] \geq (1 - O(\epsilon)) \cdot \opt$.

In order to analyze $\E[\alg_{ro}]$,  we define an event $\bev$ which corresponds to the event that for some $t \in [T/2]$ and $j \in [m]$:
$\price_{t+1, j} = \initprice \cdot \exp\left(\epsilon \cdot \left(\sum_{s\leq t}\calg_{s, j} - t \cdot a^*_j \right)\right) > \initprice \cdot  (m\beta/\epsilon)^8$; that is, termination condition of \Cref{alg:ropip} is satisfied.

We analyze $\E[\alg_{ro}]$ by thresholding $\pr(\bev)$. When $\pr(\bev)$ is large, the following lemma guarantees that our algorithm achieves a good value:

\begin{Lemma}
    \label{lma:bevlarge}
    If $\pr(\bev) \geq \epsilon$, then $\alg_{ro} \geq \opt$.
\end{Lemma}

The proof  of \Cref{lma:bevlarge} is based on the simple idea  that our algorithm gets a  large revenue (more than $\opt /\epsilon$)  when $\bev$ happens, as the price grows extremely high. Hence, we defer the  proof to \Cref{sec:lmabevlarge-proof}.

\Cref{lma:bevlarge} suggests that $\alg_{ro}$ is sufficiently large when $\pr(\bev)$ is large. It remains to discuss the case where $\pr(\bev) < \epsilon$.
In this case, the algorithm reaches the half-time with probability at least $1-\epsilon$. We will show 
that $\E[\alg_{ro}]$ is approximately $\big(1 - O(\epsilon)\big)\cdot \frac{\opt}{2}$. 

Let $\stoptime$ be a random variable denoting
the time that the algorithm stops. 
Under event $\bev$, we set $\stoptime$ to be the request after which the algorithm terminates. 
Otherwise, we set $\stoptime = T/2$ and have $\pr[\stoptime = T/2] = 1 - \pr(\bev)$.
Furthermore, let $\{\widetilde {\theta}_{t}\}_{t \in [T]}$ denote random decisions 
that follow the offline optimum decision with probability $1-4\epsilon$, i.e.,  we define 
\begin{align*}
    \widetilde \theta_t ~:=~ \left\{ 
    \begin{aligned}
        & \theta^{*}_i & \qquad & \text{with probability } ~1-4\epsilon~ \text{ when } \pi(i) = t,\\
        & \phi &\qquad  &\text{ if } t \notin G\text{ or } \text { or w.p. } 4\epsilon \text{ when } t \in G
    \end{aligned}
    \right.
\end{align*}
where $\{\theta^*_i\}$ denotes the optimal decisions made by the offline optimum on the green requests; that is, $\opt ~=~ \sum_{i = 1}^{\numg} \val _i(\theta^{*}_i).$

We start from decomposing $\alg_{ro}$ into utility and revenue:
\begin{align*} 
 \alg_{ro} ~=~ \textstyle\sum_{t = 1}^\stoptime \val_t(\theta_t) \, &= \,\, \underbrace{\textstyle\sum_{t = 1}^\stoptime \Big( \val_t(\theta_t) - \langle\bprice_t, a_t(\theta_t)\rangle \Big)}_{\mathsf{Utility}} \,\, + \,\, \underbrace{\textstyle\sum_{t = 1}^\stoptime \langle\bprice_t, a_t(\theta_t)\rangle}_{\mathsf{Revenue}} \\
 &\textstyle\geq~ \sum_{t = 1}^\stoptime \Big( \val_t(\widetilde{\theta}_{t}) - \langle\bprice_t, a_t(\widetilde{\theta}_{t})\rangle \Big) +  \sum_{i = 1}^\stoptime \langle\bprice_t, a_t(\theta_t)\rangle, \notag 
\end{align*}
where the inequality follows from the fact that decisions $\{\theta_t\}_t$ are performing best response. Taking the expectations for both sides of the above inequality and rearrange the terms, we have
\begin{align*}
   \textstyle  \E[\alg_{ro}] ~\geq~ \E\left[\sum_{t = 1}^\stoptime  \val_t(\widetilde{\theta}_{t})   \right] -  \E\left[\sum_{i = 1}^\stoptime \langle\bprice_t, a_t(\widetilde{\theta}_t) - a_t({\theta}_{t})\rangle \right] \enspace .
\end{align*}
The first term on the right hand side can be lower bound as
\[
 \E\Big[ \sum_{t = 1}^\stoptime \val_t(\widetilde{\theta}_t)\Big] ~=~ \sum_{t = 1}^{T/2} \E\Big[\val_t(\widetilde{\theta}_t) \cdot \one[\stoptime \geq t] \Big]
    ~=~\sum_{t =  1}^{T/2} \E\Big[\val_t(\widetilde{\theta}_t)\Big] \cdot \pr[\stoptime \geq t] \enspace ,
\]
where we use independence to obtain the second equality
since $\stoptime$ only depends on the history from $1$ to $t - 1$ and $\val_t(\widetilde{\theta}_t)$ depends only on the type of $t$-th request since it depends on the offline solution $\widetilde \theta_t$. This can be further simplified to obtain
\begin{align}
    \sum_{t =  1}^{T/2} \E\left[\val_t(\widetilde{\theta}_t)\right] \cdot \pr[\stoptime \geq t]&\geq~ \sum_{t =  1}^{T/2} (1 - \pr(\bev)) \cdot \E\left[\val_t(\widetilde{\theta}_t)\right] \notag \\
    &\geq~ (1 - \pr(\bev))\cdot (1 - 4\epsilon)\cdot\sum_{t =  1}^{T/2} \E\left[\val_t({\theta}^*_t)\right] \enspace . \label{eq:ro-case2-p1-1} 
 \end{align}

Finally, in order to bound $\sum_{t =  1}^{T/2} \E\left[\val_t({\theta}^*_t)\right]$ in terms of $\opt$, we bound the probability that a green request, say $\gamma_i $, contributes to the sum.
Towards this end, fix a green request $\gamma_i $, and note that
by the definition of $\widetilde \theta_t$,
this request contributes to the sum when $\pi(i) \leq T/2$. Note that
\[
\pr\left(\pi(i) ~>~ \frac{T}{2}\right) ~\leq~ \frac{{T}/{2}}{T - \numr} ~\leq~ \frac{1}{2\cdot(1-\epsilon)} \leq \frac{1}{2} + \epsilon,
\]
since there are $\numr$ red requests, and in the worst-case, they all occur in the first $T/2$ time intervals. The second inequality above uses $T \geq \numr/\epsilon$. So, we have
\begin{align}
   \sum_{t =  1}^{T/2} \E\left[\val_t({\theta}^*_t)\right] ~&=~ \sum_{i \in \numg} \val_i (\theta_i^{*}) \cdot \pr\left(\pi(i) \leq T/2\right) \geq~ \sum_{i \in \numg} \val_i (\theta_i^{*}) \cdot \left(\frac{1}{2} - \epsilon\right) ~=~ \left(\frac{1}{2} - \epsilon\right) \cdot \opt. \label{eq:ro-case2-p1-2}
\end{align}
Combining \eqref{eq:ro-case2-p1-1} and \eqref{eq:ro-case2-p1-2}, we get 
\[
\E\left[ \sum_{t = 1}^\stoptime \val_t(\widetilde{\theta}_t)\right] \geq  (1 - \pr(\bev))\cdot (1 - 4\epsilon)\cdot \left(\frac{1}{2} - \epsilon\right) \cdot \opt ~\geq~ (1 - 7\epsilon) \cdot \frac{\opt}{2},
\]
where the last inequality uses the assumption that $\pr(\bev) < \epsilon$.

So far, we have shown
\begin{equation}\label{eq:ro-part1}
    \E[\alg_{ro}] ~\geq~ (1 - 7\epsilon) \cdot \frac{\opt}{2} - {\E\Big[\sum_{t =  1}^{\stoptime} \langle \bprice_t,  a_t(\widetilde{\theta}_t) - a_t(\theta_t)   \rangle \Big]}. 
\end{equation}

As in the previous proofs, we interpret the last term in the above expression as the ``loss in revenue'' that arises from making the best-response decisions. 

\paragraph{Bounding the Loss in Revenue.}  We start by  splitting $\loss$ into two terms:
\[
\loss ~:=~ \E\Big[\sum_{t =  1}^{\stoptime} \langle \bprice_t,  a_t(\widetilde{\theta}_t) - a_t(\theta_t)   \rangle \Big]~=~ \E\Big[\sum_{t =  1}^{\stoptime} \langle \bprice_t,  a_t(\widetilde{\theta}_t) - \ba^*  \rangle \Big] + \E\Big[\sum_{t = 1}^{\stoptime} \langle \bprice_t, \ba^* - a_t(\theta_t) \rangle \Big],
\]
where we recall $\ba^* = \frac{1-\epsilon}{T} \cdot \B$ is the expected budget consumption per time step.

For the first term, note that we have $\E[a_t(\widetilde \theta_t)] \leq (1-4\epsilon) \cdot \B/(T - \numr) \leq 1 - 3\epsilon$, which is further upper-bounded by $\ba^*$. Therefore, ideally we wish to move the expectation into the dot product, and upper bound the first term by $0$. If we are in the stochastic Online Resource allocation model or if the green requests are i.i.d., the decision $\widetilde{\theta}_t$ can be made to depend only on the type of request $\gamma_t$, and is therefore independent of the price vector $\bprice_t$. However, when the green requests arrive in a random order, this random-order process brings an extra correlation between $\bprice_t$ and $\gamma_t$, preventing us from moving the expectation into the dot product. To fix this issue, we view the randomness of $\gamma_t$ as a sampling without replacement process, and further bound the first term as follows:

\begin{Lemma}
\label{lma:s-wo-r}
    We have 
    \[
    \E\Big[\sum_{t =  1}^{\stoptime} \langle \bprice_t,  a_t(\widetilde{\theta}_t) - \ba^*  \rangle \Big] ~\leq~ \epsilon \cdot \opt.
    \]
\end{Lemma}

The proof of  \Cref{lma:s-wo-r} requires concentration inequalities for sampling-without-replacement and is deferred   to \Cref{sec:lmaswor-proof}. To complete the proof of the theorem, it remains to bound
\[
\E\Big[\sum_{t = 1}^{\stoptime} \langle \bprice_t, \ba^* - a_t(\theta_t) \rangle \Big] ~=~ \sum_{j \in [m]} \E\Big[\sum_{t = 1}^\tau \lambda_{t, j} \cdot (a^*_j - \calg_{t, j})\Big] \enspace .
\]
We apply the no-regret property (\Cref{lem:noRegret}) of  Exponential Pricing  to every $j \in [m]$, which gives
\[
\sum_{t=1}^\stoptime \lambda_{t, j} \cdot (a^*_j - \calg_{t, j}) ~\leq~ \frac{2 \initprice}{\epsilon} + 4\epsilon \cdot \sum_{t = 1}^\stoptime \lambda_{t, j} \cdot \calg_{t,j}
\]
for every realization of the instance (note the change in sign). Taking the expectation and summing over all $j \in [m]$, we have 
\begin{align*}
    \E\Big[\sum_{t = 1}^{\stoptime} \langle \bprice_t, \ba^* - a_t(\theta_t) \rangle \Big] ~&\leq~ \frac{2m \initprice}{\epsilon} + 4\epsilon \cdot \E\Big[\sum_{t = 1}^\stoptime \sum_{j = 1}^m \lambda_{t,j} \cdot \calg_{t, j} \Big] ~\leq~ \epsilon \cdot \opt + 4\epsilon \cdot \E[\alg_{ro}] \enspace ,
\end{align*}
where the second inequality uses the fact that $2m \initprice/\epsilon \leq \epsilon \cdot \esopt \leq \epsilon \cdot \opt$ and that $\alg_{ro} \geq \sum_{t \in [\stoptime]} v_t(\theta_t) \geq \sum_{t \in [\stoptime]} \sum_{j \in [m]} \lambda_{t,j} \cdot \calg_{t, j}$ since the utility is always non-negative.

Finally, applying the above inequality and \Cref{lma:s-wo-r} to \eqref{eq:ro-part1} and rearranging, we have
\[
\E[\alg_{ro}] ~\geq~ \frac{1-11\epsilon}{1+4\epsilon} \cdot \frac{\opt}{2} ~\geq~ (1 - 15\epsilon) \cdot \frac{\opt}{2} \enspace ,
\]
which completes the proof of \Cref{thm:romain}.
\end{proof}

\subsection{Proof of Lemma~\ref{lma:s-wo-r} via Sampling Without Replacement}
\label{sec:lmaswor-proof}

In this subsection, we show $\E\left[\sum_{t =  1}^{\stoptime} \langle \bprice_t,  a_t(\widetilde{\theta}_t) - \ba^*  \rangle \right] \leq \epsilon \cdot \opt$. For simplicity of notation, we define vector $(\widetilde a_{t, 1}, \cdots, \widetilde a_{t, m}) = a_t(\widetilde \theta_t)$. This simplifies our target expression as
\[
\E\Big[\sum_{t =  1}^{\stoptime} \langle \bprice_t,  a_t(\widetilde{\theta}_t) - \ba^*  \rangle \Big] ~=~ \sum_{t \in [\stoptime] \setminus R} \sum_{j = 1}^m \E\left[\lambda_{t, j} \cdot \Big(\widetilde{a}_{t,j} - a^*_j \right)\Big].
\]

Now, we fix $t \in [\stoptime] \setminus R$ and $j\in [m]$, 
and give an upper bound for $\E\left[\lambda_{t, j} \cdot \left(\widetilde{a}_{t,j} - a^*_j \right)\right]$. 
To analyze this term, we \emph{condition} on the history, say $\calH$, until (and including) time $t-1$.
This includes the random realizations and decisions made until time $t-1$.
Given history $\calH$ until time $t-1$, the price vector $\bprice_t$ is deterministic, and we get:
\[
\E\left[\lambda_{t, j} \cdot \left(\widetilde{a}_{t,j} - a^*_j \right) \mid \calH\right] ~=~ \price_{t, j} \cdot \E\left[\left(\widetilde{a}_{t,j} - a^*_j \right) \mid \calH \right]
\]
It would be more accurate to denote the price as $\bprice_t(\calH)$ to make the dependence on $\calH$ explicit: we drop this when $\calH$ is clear from context. 
We crucially prove:
\begin{equation}\label{eq:ro-key-eq}
     \sum_{\calH} \pr(\calH) \cdot \E\left[\left(\widetilde{a}_{t,j} - a^*_j \right) \mid \calH \right] ~\leq~ \frac{4}{T } \cdot \left(\frac{\epsilon}{m\beta}\right)^{10},
\end{equation}
where $\pr(\calH)$ denotes the probability that history $\calH$ occurs.

Using conditional expectations, we can write
\begin{align}
    \E\left[\price_{t, j} \cdot \left(\widetilde{a}_{t,j} - a^*_j \right)\right] ~&=~ \sum_{\calH} \pr(\calH)\cdot  \E\left[\price_{t, j}(\calH) \cdot \left(\widetilde{a}_{t,j} - a^*_j \right) \mid \calH \right] \notag \\
    ~&=~ \sum_{\calH} \pr(\calH) \cdot \price_{t, j}(\calH) \cdot \E\left[\left(\widetilde{a}_{t,j} - a^*_j \right) \mid \calH \right]. \notag
\end{align}
Moreover, since $t \leq \stoptime$, it must be the case that $\price_{t, j} \leq \initprice\cdot (m\beta/\epsilon)^8$. Using this and \eqref{eq:ro-key-eq}, we get
\begin{align}
    \E\left[\price_{t, j} \cdot \left(\widetilde{a}_{t,j} - a^*_j \right)\right] ~\leq~  \initprice \cdot \left(\frac{m\beta}{\epsilon}\right)^8 \cdot \frac{4}{T} \cdot \left(\frac{\epsilon}{m\beta}\right)^{10} ~\leq~ \frac{4\initprice}{T}\cdot \left(\frac{\epsilon}{m\beta}\right)^2.\label{eq:ro-key-term-bd}
\end{align}

Summing \eqref{eq:ro-key-term-bd} over all $t \in [\stoptime]\setminus R$ and $j \in [m]$ gives:
\begin{align*}
\E\left[\sum_{t = 1}^{\stoptime} \langle \bprice_t, a_t(\widetilde{\theta}_t) - \ba^* \rangle \right] ~=~  \sum_{t\in [\stoptime] \setminus R}~ \sum_{j = 1}^m \E\left[\price_{t, j} \cdot \left(\widetilde{a}_{t,j} - a^*_j \right) \right] ~\leq~ T\cdot m \cdot \frac{4\initprice}{T}\cdot \left(\frac{\epsilon}{m\beta}\right)^2 ~\leq~ \epsilon \cdot \opt,
\end{align*}
where the last inequality uses that $\esopt \leq \opt$.

We complete the proof of Lemma~\ref{lma:s-wo-r} by proving~\eqref{eq:ro-key-eq}.

\paragraph{Proving Equation~\eqref{eq:ro-key-eq}.}
Recall the setting of Equation~\eqref{eq:ro-key-eq}. We fix $t \in [\stoptime] \setminus R$ and $j\in [m]$, 
and condition on the history $\calH$ until (and including) time $t-1$.
We want to show:
\[
    \sum_{\calH} \pr(\calH) \cdot \E\left[\left(\widetilde{a}_{t,j} - a^*_j \right) \mid \calH \right] ~\leq~ \frac{4}{T } \cdot \left(\frac{\epsilon}{m\beta}\right)^{10}.
\]

First, we decompose the LHS as follows.
\begin{align*}
\sum_{\calH}\pr(\calH)\cdot \E\left[\left(\widetilde{a}_{t,j} - a^*_j \right) \mid \calH \right] ~&=~ \sum_{\calH}\pr(\calH) \cdot \Big(\E[\widetilde{a}_{t,j}\mid \calH] -  a^*_{j} \Big) \\
&\leq~ \sum_{\calH}\pr(\calH)\cdot\Big(\E[\widetilde{a}_{t,j}\mid \calH]\Big)\cdot \one\Big[ \E[\widetilde{a}_{t,j}\mid \calH] \geq a^*_{j} \Big] \enspace .
\end{align*}

Let $A_{\overline \calH}$ be a set denoting the remaining green 
and dummy requests the do not appear in history $\calH$ (but excluding the red requests); so, $|A_{\overline \calH}| \geq T - (t-1) - \numr \geq T/4$ since $\numr < T/\epsilon$ and $\stoptime \leq T/2$. 
Since $t \in [\stoptime] \setminus R$,  the request $\gamma_t$ is chosen uniformly at random from $A_{\overline \calH}$: for any dummy request $\widetilde{a}_{t,j}= 0$, while for any green request $i \in A_{\overline \calH}$, we have $\widehat{a}_{t, j}  = \left(a_i(\theta^*_i) \right)_j \cdot (1 - 4\epsilon) \leq \left(a_i(\theta^*_i) \right)_j$. Since the sum of $\left(a_i(\theta^*_i) \right)_j$ for all $i \in [\numg]$ is at most $B_j$, while $|A_{\overline \calH}|$ is at least $T/4$, we have
\[
\E[\widetilde{a}_{t,j}\mid \calH] ~\leq~ \frac{4B_j}{T}.
\]
for any history $\calH$. This simplifies the bound above: we obtain
\begin{align*}
\sum_{\calH}\pr(\calH)\cdot \E\left[\left( \widetilde{a}_{t,j} -   a^*_{j} \right) \mid \calH \right] 
~&\leq~ \frac{4B_j}{T}\cdot \sum_{\calH}\pr(\calH)\cdot \one\Big[ \E[\widetilde{a}_{t,j}\mid \calH] \geq a^*_{j} \Big] \\
&=~ \frac{4B_j}{T}\cdot \pr_{\calH}\Big( \E[\widetilde{a}_{t,j}\mid \calH] \geq a^*_{j} \Big).
\end{align*}
We now show that $\pr_{\calH}\Big(\E[\widetilde{a}_{t,j}\mid \calH] \geq a^*_{j} \Big) ~\leq~ \left(\frac{\epsilon}{m\beta}\right)^{10} \cdot \frac{1}{B_j}$ to finish the proof. 
We demonstrate this bound by employing concentration bounds for the sample-without-replacement process.
{Indeed, observe that we can view $A_{\overline \calH}$ being generated as follows: from the set of all green and dummy requests, pick all but $|[t-1] \setminus R |$ requests uniformly at random, without replacement. The selected requests constitute $A_{\overline \calH}$ while the remaining requests appear (uniformly at random) in the first $t-1$ slots that do not contain any red requests. }%
Formally, we first define set $A_{\phi}$ to be the set containing all green and dummy requests, and let $u = |A_{\phi}| = T - \numr \geq (1-\epsilon)\cdot T$, $v = |A_{\overline \calH}| \geq T/4$, and $\mu = (1 - 4\epsilon) \cdot \frac{1}{u} \cdot \sum_{i \in [\numg]} \left(a_i(\theta^*_i)\right)_j$ denotes the \emph{unconditional} expected allocation, where the factor $1 - 4\epsilon$ is applied because the action $\widehat \theta_t$ is set to be null with probability $4\epsilon$ even if a green request arrives.
Using $u \geq (1-\epsilon)\cdot T$ and $\sum_{i \in [\numg]} \left(a_i(\theta^*_i)\right)_j \leq B_j$, we have 
\[
\mu ~\leq~ (1-4\epsilon) \cdot \frac{B_j}{(1 - \epsilon)\cdot T} ~\leq~ (1-3\epsilon)\cdot\frac{B_j}{T}.
\]
On simplifying, we obtain:
\begin{align}
\label{eq:ro-swor-mid}
    \pr_{\calH}\Big( \E[\widetilde{a}_{t,j}\mid \calH] ~\geq~ a^*_{j} \Big)  ~&=~ \pr_{\calH}\left( \E[\widetilde{a}_{t,j}\mid \calH] ~\geq~ (1 - \epsilon)\cdot\frac{B_j}{T} \right)\notag \\
    &\leq~ \pr_{\calH}\left( \E[\widetilde{a}_{t,j}\mid \calH] - (1-3\epsilon)\cdot\frac{B_j}{T} ~>~ \frac{\epsilon\cdot B_j}{T} \right)\notag\\
    ~&\leq~ \pr_{\calH}\left(\Big| \E[\widetilde{a}_{t,j}\mid \calH] - \mu\Big| > \frac{\epsilon \cdot B_j}{T}\right).
\end{align}

{Now, we use the following corollary of Bernstein's inequality for a sampling-without-replacement process (it is obtained by replacing sums with averages in \Cref{Bernstein}) to further bound \eqref{eq:ro-swor-mid}.

\begin{Corollary}
\label{cor:swor}
    Let $Y = \{y_1, \cdots, y_u\}$ be a set of real numbers in the interval $[0, M]$. Let $S$ be a random subset of $Y$ of size $v$ and let $\bar S = \frac{1}{v}\sum_{i \in S} y_i$ be the average of $S$. Setting $\mu = \frac{1}{u} \sum_{i} y_i$, we have for every $\disceps > 0$,
    \[
    \pr\left[\left|\bar S - \mu \right| \geq \disceps \right] ~\leq~ 2\cdot\exp \left(-\frac{v\disceps^2}{M(4\mu + \disceps)} \right).
    \]
\end{Corollary}
}

Applying \Cref{cor:swor} with $M = 1$ and $\disceps = \frac{\epsilon B_j}{T}$ implies
\begin{align}
\label{eq:ro-swor-mid2}
   \pr_{\calH}\left(\Big| \E[\widetilde{a}_{t,j}\mid \calH] - \mu\Big| > \frac{\epsilon \cdot B_j}{T}\right) ~&\leq~ 2\cdot\exp \left(-\frac{v \cdot \frac{\epsilon^2 B_j^2}{T^2}}{4 \mu + \frac{\epsilon B_j}{T}}\right) \notag \\
    ~&\leq~ 2\cdot\exp \left(- \frac{\epsilon^2 B_j^2}{16 \mu T + 4 \epsilon B_j} \right) ~\leq~ 2 \cdot \exp \left( - \frac{\epsilon^2 B_j}{40} \right),
\end{align}
where the second inequality uses $v \geq T/4$, and the third inequality uses $\mu \leq (1 + 2\epsilon) \cdot \frac{B_j}{T} \leq \frac{2B_j}{T}$.
To further simplify $2 \cdot\exp \left( - \frac{\epsilon^2 B_j}{40} \right)$, we define $q_j = B_j \cdot \frac{\epsilon^2}{400 \log (m\beta/\epsilon)} \geq 3$, which holds when $B_j \geq \frac{1200 \log (m\beta/\epsilon)}{\epsilon^2}$. Then, we have  
\begin{align*}
    2 \cdot\exp \left( - \frac{\epsilon^2 B_j}{40} \right) ~\leq~  \left(\frac{\epsilon}{m\beta}\right)^{10}   \cdot 2\cdot \left(\frac{\epsilon}{m\beta}\right)^{10q_j - 10} ~&\leq~
     \left( \frac{\epsilon}{m\beta}\right)^{10} \cdot 2\cdot \left( \frac{\epsilon}{m\beta}\right)^{10q_j - 13} \cdot \frac{\epsilon^2}{\log(m/\epsilon)}\\
    &\leq~ \frac{1}{\beta} \cdot \left( \frac{\epsilon}{m\beta}\right)^{10} \cdot 2\cdot \frac{1}{800q_j} \cdot \frac{\epsilon^2}{\log(m/\epsilon)} \\
    ~&=~  \left( \frac{\epsilon}{m\beta}\right)^{10} \cdot \frac{1}{B_j} \enspace ,
\end{align*}
where the second inequality uses $q_j \geq 3$ and $\epsilon \leq 0.5$ to conclude 
$(m\beta/\epsilon)^{10 q_j - 13} \geq  (1/0.5)^{5q_j} \geq 800 q_j$.
Finally, combining the above inequality with \eqref{eq:ro-swor-mid} and \eqref{eq:ro-swor-mid2} gives
\[
\pr_{\calH}\Big( \E[\widetilde{a}_{t,j}\mid \calH] ~\geq~ a^*_{j} \Big) ~\leq~ \left(\frac{\epsilon}{m\beta}\right)^{10} \cdot \frac{1}{B_j } \enspace , 
\]
which completes the proof.

\subsection{Proof of \Cref{lma:bevlarge}}
\label{sec:lmabevlarge-proof}

We finish the section by proving \Cref{lma:bevlarge} that the algorithm has total value at least $\opt$ when $\pr(\bev)$ is at least $\epsilon$.

Note that
when event $\bev$ occurs, say after request $\stoptime$ for resource $j^*$, the resulting price $\price_{\stoptime+1, j^*} > \initprice \cdot    (m\beta /\epsilon)^8$.
Similar to the analysis of \Cref{lemma:prob-ea-small}, we lower bound the revenue obtained from the {last one unit of resource $j^*$}. 
Specifically, we first note that $\sum_{\ell \leq \stoptime}\calg_{\ell, j^*} - \stoptime \cdot a^*_{j^*} \geq {8\log(m/\epsilon)}/{\epsilon}$, which is at least $1$. 
We define $\widehat{\ell}$ to be an index such that
$\sum_{\widehat{\ell} \leq i \leq \stoptime} \calg_{i, j^*} \in [1, 2]$: such an index always exists since at most one unit of a resource can be allocated for a given request.

Conditioned on $\bev$, and using $ \val_{\ell}(\theta_{\ell}) - \langle\price_\ell, a_\ell\rangle  \geq 0$ since $\phi \in \Theta_\ell$, 
we have
\begin{align*}
\textstyle\alg_{ro} &~=~ \sum_{\ell = 1}^\stoptime \val_\ell(\theta_\ell) ~\geq~  \sum_{\ell = 1}^\stoptime\sum_{j=1}^m \price_{\ell,j} \cdot \calg_{\ell, j} 
~\geq~ \sum_{\ell = \widehat{\ell}}^\stoptime\price_{\ell,j^*} \cdot \calg_{\ell, j^*} \enspace.  \\
\intertext{Substituting $\price_{\ell,j^*}$,  this can be simplified to}
\alg_{ro} &~\geq ~ \sum_{\ell=\widehat{\ell}}^\stoptime \initprice \cdot \exp\Bigg( \epsilon \Big( \sum_{i < \ell } \calg_{i, j^*} - (\ell-1)\cdot a^*_{j^*} \Big) \Bigg) \cdot \calg_{\ell, j^*}\\
&~\geq \sum_{\ell=\widehat{\ell}}^\stoptime \initprice \cdot \exp\Bigg( \epsilon \Big( \sum_{i \leq \stoptime } \calg_{i, j^*} - \stoptime\cdot a^*_{j^*} - \sum_{\ell \leq i \leq \stoptime} \calg_{i, j^*} \Big) \Bigg) \cdot \calg_{\ell, j^*} \enspace .
\end{align*}
Now using the termination condition and the fact that $1 \leq \sum_{\widehat{\ell} \leq i \leq \stoptime} \calg_{i, j^*} \leq 2$, we get
\begin{align*}
\alg_{ro}  
&~\geq~ \sum_{\ell=\widehat{\ell}}^{\stoptime}  \initprice \cdot \exp\Bigg( \epsilon \Big( \frac{8\log(m\beta/\epsilon)}{\epsilon} - 2 \Big) \Bigg) \cdot \calg_{\ell, j^*} ~\geq~ \initprice\cdot\left(\frac{m\beta}{\epsilon}\right)^6. 
\end{align*}
Hence, if $\pr(\bev) \geq \epsilon$ then
\begin{equation*}
    \E[\alg_{ro}] ~\geq~ \pr(\bev) \cdot \E[\alg_{ro} \mid \bev] ~\geq~ \epsilon \cdot\initprice\cdot\left(\frac{m\beta}{\epsilon}\right)^6 ~\geq~ \opt \enspace ,
\end{equation*}
where in the last inequality we use the fact that $\beta \cdot \esopt \geq \opt$.

\section{Robustness in the Prophet-with-Augmentations Model}
\label{sec:sub-augmentation}

In this section, we show that the Exponential Pricing algorithm 
is robust to augmentations,
and obtains a $(1-\epsilon)$ approximation in the 
Prophet-with-Augmentations model of \cite{immorlica2020prophet,AGMS-SODA22}.

\medskip
\noindent{\bf Prophet-with-Augmentations Model.} We are given a base instance of 
the general online resource allocation problem in the stochastic model from \Cref{sec:introModel} with \emph{known} distributions.
That is,  $n$ requests  arrive sequentially, where the $i$-th request $\gamma_i = (\val_i, a_i, \Theta_i)$ is independently drawn from a known distribution $\calD_i$.
The distribution $\calD_i$ is a distribution over $K$ request tuples
$\gamma_{i,1}, \ldots, \gamma_{i,K}$, where the $i$-th request is $\gamma_{i,k} = (\val_{i,k}, a_{i,k}, \Theta_{i,k})$ with probability $p_{i, k}$.
An \emph{oblivious} adversary fixes 
perturbation functions $\{r_{i, k}\}_{i \in [n], k \in [K]}$ such that 
$r_{i,k}: \Theta_{i,k} \to \R_{\geq 0}$. When request $i$ arrives with type $k$, the algorithm observes a perturbed request $(\val_{i, k} + r_{i, k}, a_{i,k}, \Theta_{i,k})$.

To simplify notation, let the $i$-th request be $(\val_i+r_i, a_i, \Theta_i)$. After our posted pricing algorithm observes the perturbed request, it makes a best-response decision $\theta_i \in \Theta_{i}$. 
The algorithm receives the augmented value $\val_{i}(\theta_i) + r_{i}(\theta_i)$, and the
goal is to maximize total value, $\sum_{i = 1}^n \big( \val_i(\theta_i)+r_i(\theta_i) \big)$, subject to the budget constraints.
Crucially, the algorithm competes against the expected offline optimum of the base instance; that is, when there are no perturbations. 

The following is the main result of this section.

\begin{Theorem}\label{thm:onlineRA-aug}
    Given an $\epsilon > 0$, there exists a pricing algorithm  for online resource allocation in the Prophet-with-Augmentations model that obtains
    expected total value at least $\left(1 - O(\epsilon)\right)$ times the expected hindsight optimum with no perturbations, provided that $B_{j} = \Omega\left(\frac{\log(nm/\epsilon)}{\epsilon^2}\right) \, \forall j \in [m]$.
\end{Theorem}

\begin{proof} We assume $\epsilon \in [0,1/2]$ to be a fixed  parameter, since the theorem is true for $\epsilon>1/2$.

 Our main idea is to run \Cref{alg:prsp} for the base instance. Since we know the base instance distributions, it's sufficient to run \Cref{alg:prsp} with the accurate value of $\opt$ and the optimal solution $x^*$ of $\lp((1 - \epsilon) \cdot \B)$, i.e., the estimates $\esopt$ and $\{\widehat a_{i,j}\}$ become $\opt$ and $\{a^*_{i,j}(x^*)\}$ respectively.

The proof is almost identical to the proof of \Cref{thm:exp-pricing-estimates}. We now describe the small changes needed to handle the adversarial augmentations. 
The only change to Algorithm~\ref{alg:prsp} is that the decision $\theta_i$ is computed as 
\begin{equation*}
\theta_i = \arg\max_{\theta \in \Theta_i} \Big( \val_i(\theta) + r_i(\theta) - \langle\bprice_i, a_i(\theta)\rangle \Big) \enspace.
\end{equation*}

In the proof, we still define $\EA$ to be the event that the budget constraint $ \sum_{\ell \leq i} \calg_{\ell, j} \,\, \geq \,\, \sum_{\ell \leq i}\cons^*_{\ell, j}(x) + \frac{1}{2}\epsilon\cdot B_j$ is violated. Since the adversary is oblivious, the probability $\pr(\EA)$ is still well defined. However, instead of bounding $\pr(\EA)$, we need to slightly change the argument by showing that our algorithm gains a value of at least $10\opt$ when $\pr(\EA) \geq \epsilon$. This is true because the proof of \Cref{lemma:prob-ea-small} implies that our algorithm gains a value of at least $\frac{10}{\epsilon} \cdot \opt$ when $\EA$ happens.

For the case that $\pr(\EA) < \epsilon$, the changes in the analysis are mainly regarding the best response decision: since our algorithm picks $\theta_i$ that maximizes $\val_i(\theta) + r_i(\theta) - \langle \bprice_i, a_i(\theta) \rangle$, we have
\[
\val_i(\theta_i) + r_i(\theta_i) - \langle \bprice_i, a_i(\theta_i) \rangle ~\geq~ \val_i(\theta_i^*) + r_i(\theta_i^*) - \langle \bprice_i, a_i(\theta_i^*) \rangle ~\geq~ \val_i(\theta_i^*)  - \langle \bprice_i, a_i(\theta_i^*) \rangle,
\]
where the inequality uses the fact that $r_i$'s are non-negative. 
By summing over all $i \in [\tau]$, where random variable $\tau$ denotes after which the algorithm stops, and taking expectations, we get
\[
\E[\alg] ~=~ \E\Big[\sum_{i = 1}^\tau (\val_i(\theta_i) + r_i(\theta_i))\Big] ~\geq~ \E\Big[ \sum_{i = 1}^\tau \val_i(\theta^*_i)\Big] + \E\Big[\sum_{i = 1}^{\tau} \langle \bprice_i, a_i(\theta_i) - a_i(\theta^*_i)  \rangle \Big].
\]
Note that the above inequality is identical to \eqref{eq:for-augmentation}, and so is the remainder of the analysis.

Combining the two cases that $\pr(\EA) \geq \epsilon$ and $\pr(\EA) < \epsilon$ proves \Cref{thm:onlineRA-aug}. 
\end{proof}

\bigskip
\bigskip

\appendix

\section{Basic Probabilistic Inequalities}

\begin{Theorem}[Hoeffding's Inequality]
\label{Hoeffding}
Let $X_1,\ldots,X_N$ be independent random variables such that $a_i \leq X_i \leq b_i$. Let $S_N = \sum_{i \in [N]} X_i$. Then, for all $t > 0$, we have
$  \Pr{|S_N - \Ex{S_n}| \geq t} \leq 2\exp\big(-\frac{2t^2}{\sum_{i \in [n]} (b_i - a_i)^2}\big).$
This implies, that if $X_i$ are i.i.d.\,samples of random variable $X$, and $a = a_i, b = b_i$ for all $i \in [N]$, let  $\hat X := \frac{1}{N} \sum_{i \in [N]} X_i$, then for any $\varepsilon > 0$, we have
\[
    \pr\left(|\hat{X} - \Ex{X}| \geq \varepsilon\right) \leq 2 \exp \left(-\frac{2N\varepsilon^2}{(b-a)^2}\right).
\]
\end{Theorem}

\begin{Theorem}[Bernstein's Inequality for Bounded Variables]
\label{Bernstein-bounded}
Let $X_1, \ldots, X_N$ be independent mean-zero random variables such that $|X_i| \leq M$ for all $i$ and $\sigma^2 := \sum_{i \in [N]} \E[X^2_i]$. Then, for any $\varepsilon \geq 0$, 
\[
\pr \left(\sum_{i = 1}^N X_i \geq \varepsilon \right) \leq \exp\left(- \frac{\varepsilon^2/2}{\sigma^2 + M\varepsilon/3}\right) \quad \text{and} \quad \pr \left(\sum_{i = 1}^N X_i \leq -\varepsilon \right) \leq \exp\left(- \frac{\varepsilon^2/2}{\sigma^2 + M\varepsilon/3}\right) \enspace .
\]
\end{Theorem}

We also require the following Bernstein’s inequality for sampling-without-replacement:
\begin{Theorem}[see Corollary 2.4.2 in \cite{GM-MOR16} and Theorem 2.14.19 in \cite{weakConvergence}]
\label{Bernstein}
    Let $Y = \{y_1, \cdots, y_u\}$ be a set of real numbers in the interval $[0, M]$. Let $S$ be a random subset of $Y$ of size $v$ and let $Y_S = \sum_{i \in S} y_i$. Setting $\mu = \frac{1}{u} \sum_{i} y_i$, 
    we have that for every $\tau > 0$,
    \[
    \pr\left[\left|Y_s - v \cdot \mu \right| \geq \tau \right] ~\leq~ 2\exp \left(-\frac{\tau^2}{M(4v\mu + \tau)} \right).
    \]
\end{Theorem}

\section{Missing Proofs from \Cref{sec:single-sample}}

\subsection{Proof of \Cref{clm:nottoomuch}}\label{app:clm-nottoomuch}

\clmnottoomuch*

\begin{proof}
Let $\ptt = \{S_1, \cdots, S_D\}$ be a partition chosen uniformly at random. 
To prove the claim, it suffices to show that for each group $S_d \in \ptt$, 
\[
\textstyle \sum_{i \in S_d} \sum_{k \in [K]} \sum_{\theta \in \Theta_{i,k}} a_{i,k}(\theta) \cdot \tildx_{i,k,\theta} ~\leq~ (1 - 2\epsilon) \cdot \B/D 
\]
holds with probability $1 - \epsilon/D$. Applying the union bound over all groups completes the proof.

Fix $j \in [m]$. Let variable $X_{i, j} = \sum_{k \in [K]} \sum_{\theta \in \Theta_{i,k}} \left(a_{i,k}(\theta)\right)_j \cdot \tildx_{i,k, \theta}$. 
Then, we have $\sum_{i \in [n]} X_{i,j} \leq (1 - 3\epsilon) \cdot B_j$, which is guaranteed by the fact that $\tildx$ is a feasible solution for $\lp((1 - 3\epsilon) \cdot \B)$, and $X_{i,j} \in [0, 1]$, which follows from the fact that each request uses at most one unit of resource $j$. 
Note that the sum $\sum_{i \in S_d} \sum_{k \in [K]} \sum_{\theta \in \Theta_{i,k}} \left(a_{i,k}(\theta)\right)_j \cdot \tildx_{i,k, \theta}$ can be viewed as uniformly choosing $|S_d| = n/D$ samples from the pool $\{X_{i,j}\}_{i \in [n]}$ without replacement. Applying \Cref{Bernstein}, we get 
\begin{align}\textstyle
    \pr \Big(\sum_{i \in S_d} X_{i, j} - (1-3\epsilon) \cdot B_j/D \geq \epsilon \cdot B_j/D \Big) ~\leq~ \exp \left(-\frac{\epsilon^2 B^2_j/D^2}{1 \cdot (4B_j/D + \epsilon \cdot B_j/D)}\right) ~\leq~ \frac{\epsilon}{mD}, \label{eq:partitioncond}
\end{align}
where the last inequality holds when $B_j \geq \frac{10\log (mD/\epsilon) \cdot D}{\epsilon^2} \geq \frac{10\log (mn/\epsilon) \cdot D}{\epsilon^2}$, as there must be $D \leq n$.
Applying the above inequality to every $j \in [m]$ and combining them via union bound implies that \eqref{eq:partitioncond} holds simultaneously for every $j \in [m]$ with probability at least $1 - \epsilon/D$, as desired.
\end{proof}

\subsection{Proof of \Cref{clm:followtildx}}\label{app:clm-followtildx}

\clmfollowtildx*

\begin{proof}
Fix $S_d \in \ptt$, and let
$\{\gamma_i\}_{i \in S_d}$ denote the requests. 
Recall that $\tildz_{i, \theta}^{(d)} =\frac{\tildx_{i,k, \theta}}{p_{i,k}} \cdot \one[\type(i)=k]$, if Condition~\eqref{eq:condition} is satisfied, which requires:
\begin{align*}
\textstyle
    \sum_{i \in S_d} \sum_{\theta \in \Theta_i} a_i(\theta) \cdot \sum_{k \in [K]} \one[\type(i) = k] \cdot  \frac{\tildx_{i,k, \theta}}{p_{i,k}} \leq (1 - \epsilon) \cdot \B/D .
\end{align*}
Therefore, we have
\begin{align*}
\textstyle
\E_{\{\gamma_i\}} \Big[\sum\limits_{i \in S_d} \sum\limits_{\theta \in \Theta_i} v_i(\theta) \cdot \tildz^{(d)}_{i, \theta} \Big]
&=~ \textstyle{ \E_{\{\gamma_i\}}\Big[\sum\limits_{i \in S_d} \sum\limits_{k \in [K]} \sum\limits_{\theta \in \Theta_{i,k}} v_{i,k}(\theta) \cdot \frac{\tildx_{i, k, \theta}}{p_{i,k}} \cdot \one[\type(i) = k] \cdot \one[\eqref{eq:condition} \text{ holds}] \Big]} \\
~& \textstyle =~ \sum\limits_{i \in S_d}\sum\limits_{k \in [K]} \sum\limits_{\theta \in \Theta_{i,k}} v_{i,k}(\theta) \cdot\frac{\tildx_{i, k, \theta}}{p_{i,k}} \cdot \E_{\{\gamma_i\}}\big[  \one[\type(i) = k] \cdot    \one[\eqref{eq:condition} \text{ holds}]\big]
\end{align*}

We note that the events $\one[\type(i) = k]$ and $\one[\eqref{eq:condition} \text{ holds}]$ can be correlated, which prevents us from further decomposing the expectation above. 
To address this, we introduce the following condition for every $i \in S_d$.
\begin{align}
\textstyle
    \sum_{\ell \in S_d \setminus \{i\}} \sum_{k \in [K]} \sum_{\theta \in \Theta_\ell} \one[\type(\ell) = k] \cdot a_\ell(\theta) \cdot \frac{\tildx_{i, k, \theta}}{p_{\ell,k}} ~\leq~  (1 - 1.5\epsilon) \cdot \B/D. \label{eq:conditionnew}
\end{align}
Crucially, we note Condition \eqref{eq:conditionnew} implies  Condition \eqref{eq:condition}. 
This is because request $i$ can consume at most one unit of each resource, and we have $0.5 \epsilon \cdot B_j/D \geq 1$, which is true as long as $B_j \geq 2D/\epsilon$. 
Using the fact that Condition \eqref{eq:conditionnew} implies  Condition \eqref{eq:condition} and that the event $\one[\type(i) = k]$ and $\one[\eqref{eq:conditionnew} \text{ holds}]$ are independent, we get
\begin{align*}
\textstyle
    \E_{\{\gamma_i\}} \Big[\sum\limits_{i \in S_d} \sum\limits_{\theta \in \Theta_i} v_i(\theta) \cdot \tildz^{(d)}_{i, \theta} \Big]
    ~& \textstyle \geq~ \sum\limits_{i \in S_d} \sum\limits_{k \in [K]} \sum\limits_{\theta \in \Theta_{i,k}} v_{i,k}(\theta) \cdot \frac{\tildx_{i, k, \theta}}{p_{i,k}} \cdot \pr\left(\type(i) = k\right) \cdot \pr(\eqref{eq:conditionnew} \text{ holds}) \\
    & \textstyle =~\sum\limits_{i \in S_d}\sum\limits_{k \in [K]} \sum\limits_{\theta \in \Theta_{i,k}} v_{i,k}(\theta) \cdot \tildx_{i, k, \theta} \cdot \pr(\eqref{eq:conditionnew} \text{ holds}).
\end{align*}
Thus, it suffices to show that Condition~\eqref{eq:conditionnew} holds with probability at least $1-\epsilon$ to complete the proof of \Cref{clm:followtildx}.

Towards this end, we fix $j \in [m]$ and let random variable $X^{(d)}_{i, j}$ denote the amount of resource $j$ consumed by request $i \in S_d$, assuming request $i$ follows solution $\tildx$. 
That is, $X^{(d)}_{i, j}$ equals $\sum_{\theta \in \Theta_{i, k}} \left(a_{i,k}(\theta) \right)_j \cdot {\tildx_{i, k, \theta}}/{p_{i, k}}$ with probability $p_{i, k}$. 
Furthermore, let $\mu^{(d)}_{i,j}$ denote the expectation of $X^{(d)}_{i,j}$. 
Since $a_{i, k}(\theta) \in [0, 1]^m$ and $\ptt$ is an average partition with respect to $\tildx$, we have $X^{(d)}_{i, j} \in [0, 1]$ and $\sum_{i \in S_d} \mu^{(d)}_{i,j} \leq (1- 2\epsilon) \cdot B_j/D$.
We will apply Bernstein's inequality (\Cref{Bernstein-bounded}) to the mean-zero random variables $\{X^{(d)}_{i,j} - \mu^{(d)}_{i,j}\}_{i \in S_d}$. 
Let $\sigma^2 = \sum_{i \in S_d} \E[(X^{(d)}_{i,j} - \mu^{(d)}_{i,j})^2]$, and note that
\[\textstyle
\sigma^2 ~=~ \sum_{i \in S_d} \E[(X^{(d)}_{i,j})^2] - \sum_{i \in S_d} (\mu^{(d)}_{i,j})^2 ~\leq~ \sum_{i \in S_d} \E[X^{(d)}_{i,j} \cdot 1] - 0 ~\leq~ \sum_{i \in S_d} \mu^{(d)}_{i, j}.
\]
Then, Bernstein's inequality together with the fact that $\sum_{i \in S_d} \mu^{(d)}_{i, j} \leq (1 - 2\epsilon) \cdot B_j/D$ gives
\[\textstyle
\pr\Big(\sum_{\ell \in S_d \setminus \{i\}} X^{(d)}_{\ell, j} - (1-2\epsilon) \cdot \frac{B_j}{D} > 0.5\epsilon \cdot \frac{B_j}{D} \Big) ~\leq~ \exp\Big(-\frac{\epsilon^2 B_j^2/(4D^2)}{\sigma^2 + 1\cdot 0.5\epsilon B_j/(3D)}\Big) ~\leq~ \frac{\epsilon}{m},
\]
where the last inequality uses $\sigma^2 \leq \sum_{i \in S_d} \mu^{(d)}_{i, j} \leq B_j/D$, and it holds when $B_j \geq \frac{4 \log(m/\epsilon) \cdot D}{\epsilon^2}$. 

Finally, by applying the union bound over all $j \in [m]$, we can conclude that Condition \eqref{eq:conditionnew} holds with probability at least $1 - \epsilon$.
\end{proof}

\section{Lower Bound for Prophet Model}
\label{sec:hardnessExample}

We prove our lower bound on the budgets needed to obtain a $(1-\epsilon)$-approximation for online resource allocation  with \emph{known} input distributions (prophet model), as claimed in \Cref{sec:results-single-sample}.

\begin{Theorem}
\label{thm:lower}
    In online resource allocation  with known input distributions, there exist instances with
    $B_j  = (\log m)/\epsilon^2$ for all resources $j \in [m]$, for which no online algorithm can achieve a competitive ratio better than $(1 - \Omega( \epsilon))$.
\end{Theorem}

\medskip \noindent \textbf{The hard instance.} The instance follows the instance given in \cite{AWY-OR14} for the random-order model.
We consider the setting of online combinatorial auctions, a special case of online resource allocation, where a sequence of $n$ buyers arrive for $m$ kinds of items/resources. 
For simplicity in describing the hard instance, we assume that $m = 2^z$, and number items from $0$ to $m - 1$. 
Each item has $B = \log m/\epsilon^2$ copies, implying $\epsilon = \sqrt{z/B}$.

In our hard instance, each buyer is \emph{single-minded}: buyer $i$ has a valuation function given by $f_i(S) = c_i \cdot \mathbf{1}[S_i \subseteq S]$, where $S_i$ is a fixed and $c_i \in \R_{\geq 0}$ could be random (i.e., they have value $c_i$ if they receive a superset of $S_i$). Each buyer belongs to one of the following \emph{types}, indexed by $l \in [z]$:
\begin{itemize}
    \item Type $A_l$: interested in the $m/2$ items $j \in \{0, \dots, m - 1\}$ whose $l$-th bit is $1$.
    \item Type $B_l$: interested in the $m/2$ items $j \in \{0, \dots, m - 1\}$ whose $l$-th bit is $0$.
\end{itemize}
Exactly $3B + \sqrt{Bz}$  buyers arrive sequentially, divided into 3 groups (with lower-numbered groups arriving earlier; the intra-group order is irrelevant for the hardness).

\begin{itemize}
    \item Group 1: $\sqrt{B/z}$ buyers of type $A_l$ arrive, for each $l \in [z]$, each with $c_i = 2$.
    \item Group 2: $2B/z$ buyers of type $A_l$ arrive, for each $l \in [z]$; the valuation $c_i$ is randomly chosen to be either $1$ or $3$, each with a probability of $0.5$.
    \item Group 3: $B/z$ buyers of type $B_l$ arrive, for each $l \in [z]$, each with $c_i = 4$.
\end{itemize}

Let $\opt$ denote the expected value achieved by the optimal offline algorithm for this instance.
We begin with a couple of observations.

\begin{Observation}\label{fac:opt-range}
    We have $5B \leq \opt \leq 7B$.
\end{Observation}
\begin{proof}

    First, note that for each $l \in [z]$, buyers of type $A_l$ and $B_l$ are interested in disjoint sets of items.
    Thus, for every realization of the instance, we can allocate items to all $B/z$ buyers of type $B_l$ and to any $B/z$ buyers of type $A_l$, for all $l \in [z]$.  
    This results in an allocation with $B$ buyers receiving value $4$ and another $B$ buyers receiving value  at least $1$, implying that  $\opt \geq 5B$.

    For the upper bound, note that the instance has a total of $B \cdot m$ items, and each buyer is interested in exactly $m/2$ items. 
    Therefore, any allocation can serve at most $2B$ buyers. Since at most $B$ of these can have value $4$ and the rest have  value at most $3$, it follows that $\opt \leq 7B$.
\end{proof}

\begin{Observation} \label{fac:common}
    For every subset $Y \subseteq [z]$, there is an item that is of interest to every buyer of type $A_l$ with $l \in Y$ and every buyer of type $B_{l'}$ with $l' \notin Y$.
\end{Observation}
\begin{proof}
    Consider item $j$ such that the $l$-th bit of $j$ is $1$ and the $l'$-th bit is 0, where $l \in Y$ and $l' \notin Y$. Then, by the definition of buyer types, \Cref{fac:common} follows.
\end{proof}

\smallskip \noindent \textbf{Online algorithms with more power.}  To prove \Cref{thm:lower}, we show that no online algorithm can achieve an expected reward greater than $(1 - 10^{-3} \cdot \epsilon) \cdot \opt$ for the above instance, even if we give the online algorithm more power.
Specifically, we assume that the online algorithm is able to make decisions based on the following \emph{two-stage arrival}:
\begin{itemize}
    \item Stage 1: All buyers in Group 1 arrive simultaneously. The algorithm needs to immediately decide how many buyers of type $A_l$ and value $2$ to pick. Let $q_l$ denote the number of type $A_l$ buyers the algorithm picks in this stage.
    
    \item Stage 2: All buyers in Groups 2 and 3 arrive simultaneously. After fixing $\{q_l\}_{l \in [z]}$ and the randomness in Group 2, the optimal allocation can be solved offline. Let $r_l$ be the number of type $A_l$ buyers the algorithm picks in this stage.
\end{itemize}

\noindent Note that Stage~1 is deterministic. We first show that the optimal allocation for Stage~2 must satisfy the following property.
\begin{Claim}
\label{clm:lower-equal}
    Fix $\{q_l\}_{l \in [z]}$ and the randomness in Group 2. Then, the optimal allocation $\{r_l\}_{l \in [z]}$ in Stage~2 must satisfy $q_l + r_l = B/z$ for each $l \in [z]$.
\end{Claim}

\begin{proof}
    We first show that $q_l + r_l \leq B/z$ for each $l \in [z]$.  
    Towards a contradiction, let $Y \subseteq [z]$ such that $Y = \{l \in [z] : q_l + r_l > B/z \}$. 
    Without loss of generality, assume that $Y \neq \emptyset$.
    We begin by giving an upper bound on the number of buyers from Group 3 (i.e., buyer of type $B_l$ for $l \in [z]$) that can be served. 
    As a consequence of \Cref{fac:common}, there exists an item $j_Y$ that is of interest to every buyer of type $A_l$ and every buyer of type $B_{l'}$ where $l \in Y$ and $l' \notin Y$. Since $\sum_{l \in Y}(q_l + r_l)$ copies of $j_Y$ have already been allocated, the algorithm can serve at most $B - \sum_{l \in Y}(q_l + r_l)$ buyers of type $B_{l'}$ where $l' \notin Y$. Since there are exactly $B/z$ buyers of type $B_l$ for each $l \in [z]$, the total number of buyers from Group 3 that can be served is at most $\sum_{l \in Y}B/z + B - \sum_{l \in Y}(q_l + r_l) = B - \sum_{l \in Y}(q_l + r_l - B/z)$.

    We now modify the solution $\{r_l\}_{l \in [z]}$ as follows: we set $\hat{r}_l = B/z-q_l$ if $l \in Y$, else, $\hat{r}_l = r_l$.
    Then, the total loss (resulting from dropped buyers of type $A_l$) is at most $3 \cdot \sum_{l \in Y} (q_l + r_l - B/z)$.
    However, note that after the modification, we have $q_l + \widehat{r}_l \leq B/z$ for each $l \in [z]$. 
    This implies that buyers from Groups 1 and 2 use at most $B/2$ copies of each item. Consequently, all buyers from Group 3 (with value 4) can be served, resulting in a gain of at least $4 \cdot\sum_{l \in Y} (q_l + r_l - B/z)$. This leads to a gain of at least  $\sum_{l \in Y} (q_l + r_l - B/z) > 0$, which contradicts the optimality of $\{r_l\}_{l \in [z]}$.

    Finally, we show that in any optimal allocation, it must hold that $q_l + r_l = B/z$ for each $l \in [z]$. 
    This follows from the observation that if $q_l + r_l \leq B/z$ for each $l \in [z]$,
    then even after serving $B/z$ buyers of type $A_l$ (note Group 2 contains $2B/z$ buyers of type $A_l$) and all $B/z$ buyers of type $B_l$, the algorithm uses at most $B$ copies of each item.
    Therefore, the greedy principle guarantees that the optimal allocation must fully utilize this capacity, implying $q_l + r_l = B/z$ for each $l \in [z]$.
\end{proof}

\paragraph{Proving \Cref{thm:lower}.} 
With \Cref{clm:lower-equal}, the optimal strategy of the two-stage online algorithm is clear: allocate to all type $B_l$ buyers with value $4$ and to $B/z$ type $A_l$ buyers, for every $l \in [z]$. 
However, since buyers in Group 1 arrives first, the online algorithm cannot predict the number of type $A_l$ buyers with value $3$. 
Consequently, the algorithm is unable to select the top $B/z$ type $A_l$ buyers, incurring a loss relative to the offline optimum. To formalize this, we state the following claim, which follows from standard anti-concentration properties of the binomial distribution.

\begin{Claim}
    \label{clm:lower-anticon}
    For every $l \in [z]$, with probability at least $0.2$, there are at least $B/z + 0.2 \sqrt{B/z}$ type $A_l$ buyers with value $3$. Symmetrically, with probability at least $0.2$, there are at most $B/z - 0.2 \sqrt{B/z}$ type $A_l$ buyers with value $3$.
\end{Claim}
\begin{proof}
    For simplicity of the notation, we define $N = B/z$ and assume $N$ is an integer. Then, the distribution of type $A_l$ buyers in Group 2 with value $3$ follows from a Binomial distribution $Bin(2N, 0.5)$.  Since $Bin(2N, 0.5)$ is symmetric, to prove \Cref{clm:lower-anticon}, it's sufficient to prove that
    \begin{align}
        \mathop{\pr}\limits_{x \sim Bin(2N, 0.5)}\left[N - 0.2\sqrt{N} \leq x \leq N + 0.2 \sqrt{N}\right] ~\leq~ 0.6 \label{eq:lower-tmp}
    \end{align}
    Since the binomial coefficient of $Bin(2N, 0.5)$ is maximized when $x = N$, we have
    \begin{align}
        \mathop{\pr}\limits_{x \sim Bin(2N, 0.5)}\left[N - 0.2\sqrt{N} \leq x \leq N + 0.2 \sqrt{N}\right] ~&\leq~ \mathop{\pr}\limits_{x \sim Bin(2N, 0.5)}\left[x = N\right] \cdot (0.4\cdot \sqrt{N} + 1) \notag \\
        ~&\leq~ \frac{\binom{2N}{N}}{2^{2N}} \cdot 0.5\sqrt{N}, \label{eq:binom-coeff}
    \end{align}
    where the last inequality holds when $N$ is sufficiently large. 
    Thus, it suffices to show that
    $\binom{2N}{N} \cdot \sqrt{N} ~\leq~ 1.2 \cdot 2^{2N}$.
    We proceed by using Stirling's approximation (see \cite{stirling}) to bound $\binom{2N}{N}$. For every positive integer $N$, we have
    \[
\sqrt{2\pi N} \cdot \left(\frac{N}{e}\right)^N \cdot e^{\frac{1}{12N + 1}} ~<~ N! ~<~ \sqrt{2\pi N} \cdot \left(\frac{N}{e}\right)^m \cdot e^{\frac{1}{12N}}.
\]
Plugging the above inequalities in the equation $\binom{2N}{N} = \frac{(2N)!}{(N!)^2}$, we get
\[
\binom{2N}{N} ~<~ \frac{4^N}{\sqrt{\pi \cdot N}} \cdot \frac{e^\frac{1}{24N + 1}}{e^\frac{1}{6N}} ~\leq~ \frac{4^N}{\sqrt{\pi \cdot N}},
\]
which when plugged into \eqref{eq:binom-coeff} completes the proof of the claim. \qedhere 
\end{proof}

Now, we can complete the proof of  \Cref{thm:lower}. 
\begin{proof}[Proof of \Cref{thm:lower}]
Fix an online algorithm that has more power to make decisions based on the two-stage arrival process and satisfies \Cref{clm:lower-equal}. 
We define $X \subseteq [z]$ to be the subset of types such that $q_l  \geq 0.1 \cdot \sqrt{B/z}$. Then, either $|X| > z/2$, or $|X| \leq z/2$. 

When $|X| > z/2$, \Cref{clm:lower-anticon} implies that for every $l \in X$, with probability at least $0.2$, there are at least $B/z + 0.2 \sqrt{B/z} > B/z$ type $A_l$ buyers with value $3$. 
The optimal offline algorithm can allocate to $B/z$ buyers of type $A_l$ and value $3$, by foregoing type $A_l$ buyers from Group 1.
Consequently, the optimal offline allocation gains at least $(3 - 2) \cdot 0.1 \sqrt{B/z} = 0.1 \sqrt{B/z}$ more than the online algorithm from type $A_l$ buyers. Therefore, when $|X| > z/2$, the value of the online algorithm is at least $0.2 \cdot \frac{z}{2} \cdot 0.1\sqrt{B/z} = 0.01 \sqrt{Bz}$  less than $\opt$.

When $|X| \leq z/2$, \Cref{clm:lower-anticon} implies that for every $l \in [z] \setminus X$, with probability at least $0.2$, there are at most $B/z - 0.2 \sqrt{B/z}$ type $A_l$ buyers with value $3$. Since \Cref{clm:lower-equal} guarantees that the online algorithm allocates to $B/z$ buyers of type $A_l$, for $l \in [z] \setminus X$, the online algorithm picks at least $B/z - (B/z - 0.2 \sqrt{B/z}) - 0.1\sqrt{B/z} = 0.1\sqrt{B/z}$ type $A_l$ buyers with value $1$. 
However, the optimal offline allocation can replace each value $1$ buyer with a  buyer with value $2$, and obtain at least $(2 - 1) \cdot 0.1 \sqrt{B/z} = 0.1 \sqrt{B/z}$ more value  than the online algorithm from type $A_l$ buyers. 
Therefore, when $|X| \leq z/2$, the value of the online algorithm is also at least $0.2 \cdot \frac{z}{2} \cdot 0.1\sqrt{B/z} = 0.01 \sqrt{Bz}$  less than $\opt$.

To conclude, we show that no online algorithm can achieve an expected value more than $\opt - 0.01 \sqrt{Bz}$, even the online algorithm is equipped with more power to make decisions. 
Recall that \Cref{fac:opt-range} gives $5B \leq \opt \leq 7B$ and $\epsilon$ is defined as $\sqrt{z/B}$, which implies,
\[
\opt - 0.01\sqrt{Bz} ~\leq~ (1 - 10^{-3} \cdot \epsilon) \cdot \opt \enspace . \qedhere
\]
\end{proof}

{\small
\bibliographystyle{alpha}
\bibliography{ref.bib}
}
\end{document}